\documentclass[11pt]{article}
\usepackage{color}
\usepackage{amstext, amssymb, amsthm, amsmath,bm,graphicx,rotating,array,natbib}
\usepackage{booktabs,url}
\usepackage[margin=1.2in]{geometry}
\usepackage{float} 

\usepackage{booktabs}
\usepackage{dcolumn}
\usepackage{array}
\usepackage{enumerate} 

\newtheorem{lemma}{{\bf Lemma}}

\newtheorem{theorem}{{\bf Theorem}}

\newcommand{\real}[1]{\mathrm{I \! R} \mathit{^{#1}}}
\newcommand{\trans}{^{\mbox{\tiny {\sf T}}}}

\newcommand{\argmax}{\mathrm{argmax}}
\newcommand{\argmin}{\mathrm{argmin}}
\newcommand {\expect}{\mbox{E}}
\newcommand {\identity}{\mathrm{1}}

\newcommand{\Wbf}{{\bm W}}
\newcommand{\Xbf}{{\bm X}}

\newcommand{\xbf}{{\bm x}}

\newcommand{\zerobf}{{\mathbf 0}}
\newcommand{\onebf}{{\mathbf 1}}

\newcommand{\greekbold}[1]{\mbox{\boldmath $#1$}}

\newcommand{\betabf}{\greekbold{\beta}}

\newcommand{\gammabf}{\greekbold{\gamma}}

\newcommand{\phibf}{\greekbold{\phi}}

\newcommand{\thetabf}{\greekbold{\theta}}

\def\inprob{\,{\buildrel p \over \rightarrow}\,}
\def\indist{\,{\buildrel d \over \rightarrow}\,}

\begin{document}

\title{Robust regression for optimal individualized treatment rules} \vspace{0.1in}
\author{Wei Xiao$^\dagger$, Hao Helen Zhang$^\ddagger$, and Wenbin Lu$^\dagger$ \vspace{0.1in}\\
\small{$^\dagger$\textit{Department of Statistics, North Carolina State University, Raleigh, NC 27695}}\\
\small{$^\ddagger$\textit{Department of Mathematics, University of Arizona, Tucson, AZ 85721}}\\
\footnotesize{wxiao@ncsu.edu \quad hzhang@math.arizona.edu \quad lu@stat.ncsu.edu}}
\date{}
\maketitle

\baselineskip=20pt

\begin{quotation}
\noindent {\it Abstract:}
Because different patients may response quite differently to the same drug or treatment, there is increasing interest in discovering individualized treatment rule. In particular, people are eager to find the optimal individualized treatment rules, which if followed by the whole patient population would lead to the ``best'' outcome. In this paper, we propose new estimators based on robust regression with general loss functions to estimate the optimal individualized treatment rules. The new estimators possess the following nice properties: first, they are robust against skewed, heterogeneous, heavy-tailed errors or outliers; second, they are robust against misspecification of the baseline function; third, under certain situations, the new estimator coupled with pinball loss approximately maximizes the outcome's conditional quantile instead of conditional mean, which leads to a different optimal individualized treatment rule comparing with traditional Q- and A-learning. Consistency and asymptotic normality of the proposed estimators are established. Their empirical performance is demonstrated via extensive simulation studies and an analysis of an AIDS data.\par

\vspace{9pt}
\noindent {\it Key words and phrases:} Optimal individualized treatment rules; Personalized medicine; Quantile regression; Robust regression.
\par
\end{quotation}\par

\newpage
\section{Introduction}

Given the same drug or treatment, different patients may respond quite differently. Factors causing individual variability in drug response are multi-fold and complex. This has raised increasing interests of individualized medicine, where customized medicine or treatment is recommended to each individual according to his/her characteristics, including genetic, physiological, demographic, environmental, and other clinical information. The rule that applied in personalized medicine to match each patient with a target treatment is called individualized treatment rule (ITR), and our goal is to find the ``optimal'' one, which if followed by the whole patient population would lead to the ``best'' outcome. For many complex diseases such as cancer and AIDS, the optimal individualized treatment rule or regime is a dynamical treatment process, involving a sequence of treatment decisions made at different time points throughout the disease evolving course.

Q-learning \citep{watkins1992q,murphy2005generalization} and A-learning \citep{murphy2003optimal,robins2004optimal} are two main approaches for finding optimal dynamic individualized treatment rules based on clinical trials or observational data. Q-learning is based on posing a regression model to estimate the conditional expectation of the outcome at each time point, and then applying a backward recursive procedure to fit the model.  A-learning, on the other hand, only requires modeling the contrast function of the treatments at each time point, is therefore more flexible and robust to a model misspecification. See \cite{schulte2014q} for a complete review and comparison of these two methods under various scenarios, in terms of the parameter estimation accuracy and the estimation of expected outcomes. Q- and A-learning have good performance when model is correctly specified but are sensitive to model misspecification. To overcome this shortcoming, several ``direct'' methods have been proposed, which maximize value functions directly instead of modeling the conditional mean. See \cite{ZhaoYingQi2012OWL,Zhang2013Robust} for example.  

All existing methods for optimal individualized treatment rule estimation, including Q-learning and A-learning, belong to mean regression as they estimate the optimal estimator by maximizing expected outcomes. In the case of single decision point, Q-learning is equivalent to the least-squares regression. Least-squares estimates are optimal if the errors are i.i.d. normal random variables. However, skewed, heavy-tailed, heteroscedastic errors or outliers of the response are frequently encountered. In such situations, the efficiency of the least square estimates is impaired. One extreme example is that when the response takes i.i.d. Cauchy errors, neither Q-learning nor A-learning can consistently estimate the optimal ITR. For example, in  AIDS Clinical Trials Group Protocol 175 (ACTG175) data \citep{hammer1996trial}, HIV-infected subjects were randomized to four regimes with equal probabilities, and our objective is to find the optimal ITR for each patient based on their age, weight, race, gender and some other baseline measurements. The response CD4 count of the data follows a skewed, heteroscedastic errors, which weakens the efficiency of classical Q- and A-learning. A method to estimate optimal ITR which is robust against skewed, heavy-tailed, heteroscedastic errors or outliers is highly valuable. One possible solution is to construct the optimal decision rule based on the conditional median or quantiles of response given covariates than based on average effects.

In the following, we present a simple example where a quantile-based decision rule is more preferable than a mean-based decision rules. We use higher value of response $Y$ to indicate more favorable outcomes. Figure~\ref{fig:plot1} plots the conditional density of $Y$ under two treatments, $A$ and $B$, given a binary covariate $X$ which takes the value of male and female. Under the comparison based on conditional means, $A$ and $B$ are exactly equivalent. However, conditional quantiles provide us more insight. For the male group, the conditional distribution of response given treatment $B$ is a log-normal and skewed to the right. Therefore, treatment $B$ is less favorable when either 50\% or 25\% conditional quantile are considered. For the female group, the conditional distribution of response given treatment $A$ is a standard normal while a Cauchy distribution given treatment $B$. Therefore, if we make a comparison based on $25\%$ conditional quantile, treatment $A$ is more favorable.

\begin{figure}[ht]
	\centerline{\includegraphics[width=6in]{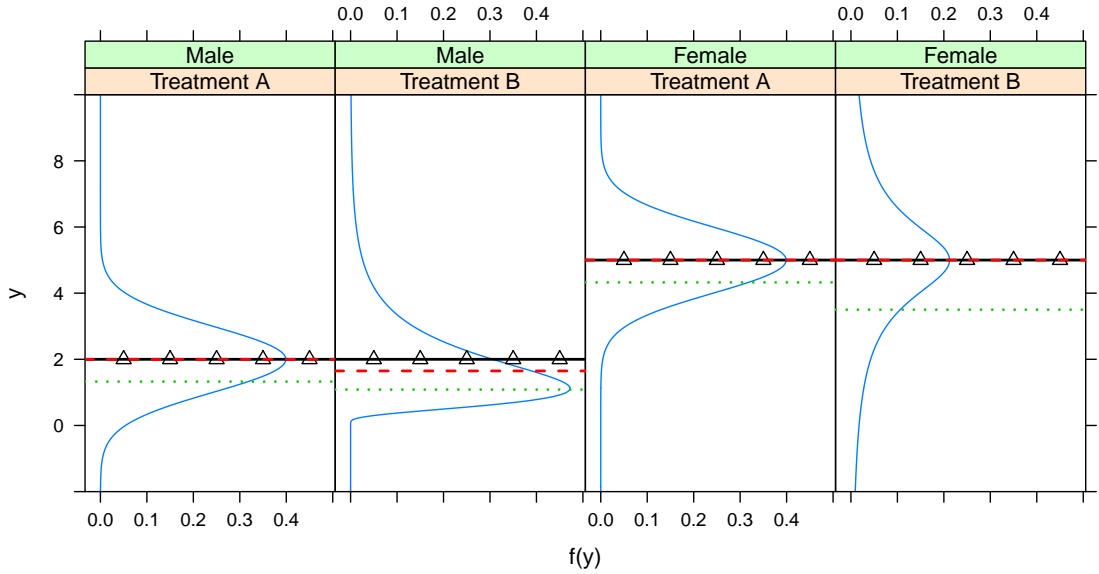}}
	\caption{The distribution functions of the response $Y$, in a randomized clinical trial with two treatments, $A$ and $B$, for male (two panels on the left) and female (two panels on the right). The solid lines with triangle symbol, dashed line, and dotted lines are the conditional mean, $50\%$ quantile, and $25\%$ quantile functions of $Y$ given the gender and the treatment, respectively.}
	\label{fig:plot1}
\end{figure}

In this paper, we propose a general framework for optimal individualized treatment rule estimation based on robust regression, including quantile regression and the regression based on Huber's loss and $\epsilon$-insensitive loss. The proposed methodology has the following desired features. First, the new decision rule obtained by maximizing the conditional quantile, which is suitable for skewed, heavy-tailed errors or outliers. Second, the proposed estimator requires only modeling the contrast function between two treatments, and is therefore robust against misspecification of the baseline function. This property is shared by A-learning. Third, empirical results from our comprehensive numerical study suggest favorable performance of the new robust regression estimator.

The rest of the paper is organized as follows. In Section 2, we first review the classical Q- and A- learning methods. Then we propose the new procedure and method and discuss its connection with existing methods. In Section 3, we study and prove the asymptotic properties of the proposed method, including consistency and asymptotic normality. In Section 4, a comprehensive numerical study is conducted to assess finite sample performance of the new procedure. In Section 5, we apply the method to ACTG175 data. Concluding remarks are given in Section 6. Throughout the paper,  we use upper case letters to denote random variables and lower case letters to denote their values.

\section{New Optimal Treatment Estimation Framework: Robust Regression}
\subsection{Basic Notations and Assumptions}
For simplicity, we consider a single stage randomized clinical trial with two treatments. For each patient, the observed data is $(\Xbf,A,Y)$, where $\Xbf\in\mathcal{X}=\real{p}$ denotes the baseline covariates, $A\in\mathcal{A}=\{0,1\}$ denotes the treatment assigned to the patient, and $Y$ is the real-valued response, which is coded so that higher values indicate more favorable clinical outcomes.  An ITR $g$ is a function mapping from $\mathcal{X}$ to $\mathcal{A}$.

We first review the potential outcome framework \citep{neyman1923applications,rubin1974estimating,rubin1986comment}. The potential outcome $Y^*(a)$ is the outcome for an arbitrary individual has s/he received treatment $a$. In actuality, at most one of the potential outcomes can be observed for any individual. The optimal ITR under mean regression, which maximizes the expected outcome, is $g^{\mathrm{opt}}_\mu=\argmax_{g\in\mathcal{G}}\expect[Y^*\{g(\Xbf)\}]$. Define the propensity score
$\pi(\Xbf)\triangleq P(A=1|\Xbf)$. Following \cite{rubin1974estimating} and \cite{rubin1986comment}, we can compute the expectation of the potential outcome under the following two key assumptions.
\begin{itemize}
	\item[(C1)] \textbf{Stable Unit Treatment Value Assumption (SUTVA):} a patient's observed outcome is the same as the potential outcome for the treatment that s/he actually received. Based on \cite{rubin1986comment}, the SUTVA assumption implies that the value of the potential outcome for a subject does not depend on what treatments other subject receive. Specifically, we can write the SUTVA assumption as
	\begin{equation}
	Y_i=Y_i^*(1)A_i+Y_i^*(0)(1-A_i),\;i=1,\ldots,n.
	\end{equation}
	This is also referred as consistency assumption.
	\item[(C2)] \textbf{Strong Ignorability Assumption:} the treatment assignment $A$ for an individual is independent of the potential outcomes conditional on the covariates $\Xbf$, i.e., $A\bot \{Y^*(a)\}_{a\in\mathcal{A}}|\Xbf$. For a randomized clinical trial, this assumption is satisfied automatically. For an observational study, as clinicians make decisions based only on all past available information, this assumption essentially assumes no unmeasured confounders.
\end{itemize}

For consistent estimation of the optimal treatment rule, we also need to assume
\begin{itemize}
	\item[(C3)] \textbf{Positivity Assumption:} $0<\pi(\xbf)<1$, $\forall \xbf\in\mathcal{X}$.
\end{itemize}

\subsection{Existing Learning Methods: Q-learning and A-learning}
Define the Q-function $Q(\xbf,a)\triangleq\expect(Y|\xbf,a)$. Under assumptions (C1)-(C2), one can show that $g^{\mathrm{opt}}_\mu(\Xbf)=\argmax_{a\in\mathcal{A}}Q(\xbf,a)=\argmax_{a\in\mathcal{A}}\expect(Y|\Xbf,A=a)$. This suggests that, in order to find $g^{\mathrm{opt}}_\mu$, we only need to estimate the conditional expectation of $Y$ given $(\Xbf,A)$. This result serves as the foundation of Q- and A-learning framework. We further define the value function $V_\mu(g)=\expect_{\Xbf}[Q\{\Xbf,g(\Xbf)\}]$ which is simply the marginal mean outcome under the ITR $g$, and $g^{\mathrm{opt}}_\mu=\argmax_gV_\mu(g)$.

Define the $\tau$-th conditional quantile of $Y$ given $(\Xbf,A)$ as $Q_{\tau}(\Xbf,A)\triangleq\inf\{y: F_{Y|\Xbf,A}(y)\geq\tau\}$. Then we define the value function based on the $\tau$-th conditional quantile as $V_{\tau-q}(g)=\expect_{\Xbf}[Q_{\tau}\{\Xbf,g(\Xbf)\}]$, which is an analog to the definition of $V_\mu(g)$. The optimal ITR which maximizes the $\tau$-th conditional quantile is then defined as
\begin{equation}
g^{\mathrm{opt}}_{\tau}(\xbf)=\underset{a\in\mathcal{A}}\argmax Q_{\tau}(\xbf,a),\;\tau\in[0,1],
\end{equation}
and $g^{\mathrm{opt}}_\tau=\argmax_gV_{\tau-q}(g)$.

Consider the general model $\expect(Y|\Xbf,A)=h_0(\Xbf)+AC_0(\Xbf)$, where $h_0(\Xbf)$ represents the baseline effect, and $C_0(\Xbf)$ denotes the contrast effect as
$$C_0(\Xbf)=\expect(Y|\Xbf,A=1)-\expect(Y|\Xbf,A=0).$$
Therefore, $g^{\mathrm{opt}}_\mu(\Xbf)=\identity\{C_0(\Xbf)>0$\}. In Q-learning, a parametric model is often employed as a working model,
\begin{equation}
\expect(Y|\Xbf,A)=h(\Xbf; \gammabf)+AC(\Xbf;\betabf),
\end{equation}
where $h(\Xbf;\gammabf)$ and $C(\Xbf;\betabf)$ are posited parametric models for $h_0(\Xbf)$ and $C_0(\Xbf)$ respectively. Commonly a linear model is assumed for simplicity and interpretability, i.e., $h(\Xbf;\gammabf)=\gammabf\trans\tilde{\Xbf}$ and $C(\Xbf;\betabf)=\betabf\trans\tilde{\Xbf}$, where $\tilde{\Xbf}=(\onebf,\Xbf\trans)\trans$. Given the observation $\{(Y_i,\Xbf_i,A_i);\;i=1,\ldots,n\}$, the Q-learning procedure estimates the parameters $(\betabf,\gammabf)$ by minimizing the squared error loss
\begin{equation}
L_{1n}(\betabf,\gammabf)=\frac{1}{n}\sum_{i=1}^{n}\left\{Y_i-h(\Xbf_i;\gammabf)-A_iC(\Xbf_i;\betabf)\right\}^2.
\end{equation}
Denote the optimized point as $(\hat{\betabf}^Q,\hat{\gammabf}^Q)$. The estimated optimal ITR based on Q-learning is then $\hat{g}^{Q}(\xbf)\triangleq\identity\{C(\xbf;\hat{\betabf}^Q)>0\}$, which is a consistent estimator of $g^{\mathrm{opt}}_\mu(\xbf)$ if both $h(\Xbf;\gammabf)$ and $C(\Xbf;\betabf)$ are correctly specified.

A-learning is a semiparametric improvement of Q-learning by modeling only the contrast function $C_0(\Xbf)$ rather than the full Q-function. This is reasonable based on the observation that the optimal ITR $g^{\mathrm{opt}}_\mu$ only depends on $C_0(\Xbf)$. By positing $C(\Xbf;\betabf)$ for the contrast function, in A-learning, one can estimate coefficients $\betabf$ by solving the following estimating equation
\begin{equation}
\sum_{i=1}^{n}\lambda(\Xbf_i)\left\{A_i-\pi(\Xbf_i)\right\}\left\{Y_i-A_iC(\Xbf_i;\betabf)-h(\Xbf_i)\right\}=0,
\label{eq:eeofAlearn}
\end{equation}
where $\lambda(\Xbf_i)$ and $h(\Xbf_i)$ are arbitrary functions, and $\lambda(\Xbf_i)$ has the same dimension as $\betabf$. Denote the solution to \eqref{eq:eeofAlearn} by $\hat{\betabf}^{A}$. If $\mathrm{var}(Y|X)$ is constant and $C(\Xbf_i;\betabf)$ is correctly specified, the optimal choices of $\lambda(\cdot)$ and $h(\cdot)$ are $\lambda(\Xbf_i;\betabf)=\partial/\partial \betabf C(\Xbf_i;\betabf)$ and $h(\Xbf_i)=h_0(\Xbf_i)$ \citep{robins2004optimal}. In practice, one may pose models, say $\pi(\Xbf_i;\phibf)$ and $h(\Xbf_i;\gammabf)$ for $\pi(\Xbf_i)$ and $h(\Xbf_i)$ respectively, and take $\lambda(\Xbf_i;\betabf)=\partial/\partial \betabf C(\Xbf_i;\betabf)$. Under randomized designs, the propensity score $\pi(\Xbf_i)$ is known. Otherwise, a logistic model can be proposed. Under the assumption that $C(\Xbf;\betabf)$ is correctly specified, the double robustness property of A-learning states that as long as one of $\pi(\Xbf;\phibf)$ and $h(\Xbf;\gammabf)$ is correctly specified, $\hat{g}^{A}(\xbf)\triangleq\identity\{C(\xbf;\hat{\betabf}^A)>0\}$ is consistent estimator of $g^{\mathrm{opt}}_\mu(\xbf)$.

Recently, \cite*{lu2011variable} propose a variant of A-learning by a loss-based learning framework. Rewrite
\begin{align*}
\expect(Y|\Xbf,A) =& h_0(\Xbf)+AC_0(\Xbf)\\
=& \varphi_0(\Xbf)+\{A-\pi(\Xbf)\}C_0(\Xbf),
\end{align*}
where $\varphi_0(\Xbf)=h_0(\Xbf)+\pi(\Xbf)C_0(\Xbf)$. Based on the expression above,  \cite{lu2011variable} propose to estimate $(\betabf,\gammabf)$ by minimizing the following loss function
\begin{equation}
L_{2n}(\betabf,\gammabf)=\frac{1}{n}\sum_{i=1}^{n}\left[Y_i-\varphi(\Xbf_i;\gammabf)-\{A_i-\pi(\Xbf_i)\}C(\Xbf_i;\betabf)\right]^2,
\label{eq:A-loss}
\end{equation}
where $\varphi(\Xbf;\gammabf)$, $C(\Xbf;\betabf)$ are proposed models for $\varphi_0(\Xbf)$ and $C_0(\Xbf)$ respectively. Denote the minimizer of \eqref{eq:A-loss} as $(\hat{\betabf}^A_{LS},\hat{\gammabf}^{A}_{LS})$. \cite{lu2011variable} show that $\hat{g}^{A}_{LS}(\xbf)\triangleq\identity\{C(\xbf;\hat{\betabf}^A_{LS})>0\}$ is a consistent estimator of $g^{\mathrm{opt}}_\mu(\xbf)$ when the propensity score $\pi(\Xbf)$ is known or can be consistently estimated from the data, and $C(\Xbf;\betabf)$ is correctly specified. We refer to this method as least square A-learning (lsA-learning).

One main advantage of the lsA-learning, compared to the classical A-learning, is its square loss, making the procedure easy to be coupled with penalized regression to achieve variable selection in high dimensional data. Specifically, \cite{lu2011variable} propose to identify important nonzero coefficients in $\betabf$ by applying an adaptive LASSO penalty to \eqref{eq:A-loss}. Under some regularity conditions, both the selection consistency and asymptotic normality of the estimator are established in \cite{lu2011variable}. The downside of lsA-learning is that one direction of the double robustness property of the classical A-learning is lost, i.e., when $\varphi(\Xbf;\gammabf)$ is correctly specified, $\betabf$ may still not be consistent if the propensity score $\pi(\Xbf)$ is not consistently estimated. Finally, it can be shown that lsA-learning and Q-learning are equivalent when $\pi(\Xbf)$ is constant and both $\varphi(\Xbf;\gammabf)$ and $C(\Xbf;\betabf)$ take the linear form (with the space of $C(\Xbf;\betabf)$ included in the space of $\varphi(\Xbf;\gammabf)$). Similar properties hold for A-learning and Q-learning \citep{schulte2014q}.

\subsection{New Proposal: Robust Regression}
Skewed, heavy-tailed, heteroscedastic errors or outliers of the response $Y$ are frequently encountered in clinical trials. It is well known that ordinary least square estimation fails to produce a reliable estimator in such situations. The immediate consequence is the efficiency loss in the estimators produced by Q-, A-, and lsA-learning. This motivates us to adopt robust regression techniques in optimal treatment regime estimation.

We consider the following additive model,
\begin{equation}
Y_i=\varphi_0(\Xbf_i)+\{A_i-\pi(\Xbf_i)\}C(\Xbf_i;\beta_0)+\epsilon_i,\; i=1,\ldots,n,
\label{eq:model_nointeraction}
\end{equation}
where $\varphi_0(\Xbf)$ is the baseline function, $C(\Xbf;\beta_0)$ is the contrast function, $\pi(\Xbf)$ is the propensity score, and $\epsilon$ is the error term which satisfies the conditional independence assumption $\epsilon\perp A|\Xbf$. We point out that the error term defined in \eqref{eq:model_nointeraction} can be very general. For example, we could take $\epsilon=\sum_{j=1}^{K}\sigma_j(\Xbf)e_j$ for any $K\geq 1$ that allows the error distribution to change with $\Xbf$, used to model heterogeneous errors, where $\sigma_j(\Xbf)$ are arbitrary positive functions and  $e_{j}\perp (A,\Xbf)$ for all $j=1,\ldots,K$. Throughout the paper, we assume $\{(Y_i,\Xbf_i,A_i,\epsilon_i),i=1,\ldots,n\}$ are i.i.d random samples of the population.

We propose to estimate $(\betabf, \gammabf)$ by minimizing
\begin{equation}
L_{3n}(\betabf, \gammabf)=\frac{1}{n}\sum_{i=1}^{n}M\left[Y_i-\varphi(\Xbf_i;\gammabf)-\{A_i-\pi(\Xbf_i)\}C(\Xbf_i;\betabf)\right],
\label{eq:A-general-loss}
\end{equation}
where $\gammabf\in\Gamma$, $\betabf\in\mathcal{B}$ and $M:\real \rightarrow [0,\infty)$ is a convex function with minimum achieved at 0. Denote the minimizer of \eqref{eq:A-general-loss} as $(\hat{\betabf}^R_{M},\hat{\gammabf}^{R}_{M})$, and the estimated ITR is then $\hat{g}^{R}_M(\xbf)\triangleq\identity\{C(\xbf;\hat{\betabf}^R_M)>0\}$. In the following, we refer the robust regression with loss function $M(x)$ as RR(M)-learning. In this article, we consider the following three types of loss functions, i.e., the pinball loss
\begin{equation}
M(x)=\rho_\tau(x)\triangleq
\begin{cases}
(\tau-1)x, &\text{if } x<0\\
\tau x, &\text{if } x\geq0
\end{cases}
\end{equation}
where $0<\tau<1$, the Huber loss
\begin{equation}
M(x)=H_\alpha(x)\triangleq
\begin{cases}
0.5x^2, &\text{if } |x|<\alpha\\
\alpha|x|-0.5\alpha^2, &\text{if } |x|\geq\alpha
\end{cases}
\end{equation}
for some $\alpha>0$, and the $\epsilon$-insensitive loss
\begin{equation}
M(x)=J_\epsilon(x)\triangleq\max(0, |x|-\epsilon)
\end{equation}
for some $\epsilon>0$. The pinball loss are frequently applied for quantile regression \citep{koenker2005quantile}, and the Huber losses and the $\epsilon$-insensitive are robust against heavy tailed errors or outliers. A dramatic difference of pinball loss, Huber loss and $\epsilon$-insensitive loss, compared with the square loss, is that they penalize large deviances linearly instead of quadratically. This property makes them more robust when dealing with responses with non-normal type of errors.

\section{Asymptotic Properties}
\label{section:asymptotic}
\subsection{Consistency of Robust Regression:  Pinball Loss}

Under the conditional independence assumption $\epsilon\perp A|\Xbf$, we have
\begin{align*}
Q(\Xbf,A)=& \varphi_0(\Xbf)+\{A-\pi(\Xbf)\}C(\Xbf;\beta_0)+\mu_\epsilon(\Xbf);\\
Q_\tau(\Xbf,A)=& \varphi_0(\Xbf)+\{A-\pi(\Xbf)\}C(\Xbf;\beta_0)+F^{-1}_{\epsilon}(\Xbf;\tau).
\end{align*}
where $\mu_{\epsilon}(\Xbf)$ and $F^{-1}_{\epsilon}(\Xbf;\tau)$ denote the mean and the $\tau$-th quantile of $\epsilon$ conditional on $\Xbf$ respectively. Therefore, in this situation, we have $g^{\mathrm{opt}}_\mu=g^{\mathrm{opt}}_{\tau}=\identity\{C(\Xbf;\betabf_0)>0\}$. In other words, the underlying ITR which maximize the population mean and $\tau$-th quantile are equivalent. For a good ITR $\hat{g}=\identity\{C(\Xbf;\hat{\betabf})>0\}$, it is reasonable to require $\hat{\betabf}$ to be a consistent estimator of $\betabf_0$. This  consistency result is first shown for the robust regression with pinball loss, which is given in Theorem 1. We allocate all the proofs into the Appendix A.

\begin{theorem}
	Under regularity conditions (A1)-(A8) in the Appendix A, if the contrast function in \eqref{eq:model_nointeraction} is correctly specified and $\pi(\xbf)$ is known, then $\hat{\betabf}^{R}_{\rho(\tau)}\inprob\betabf_0$ for all $\tau\in(0,1)$, where $\hat{\betabf}^{R}_{\rho(\tau)}$ is the solution of \eqref{eq:A-general-loss} when $M(x)=\rho_{\tau}(x)$.
\end{theorem}
\noindent \textbf{Remarks:}
\begin{enumerate}
	\item Theorem 1 doesn't assume the finiteness of $E(Y)$. Therefore it can be applied to the cases when $\epsilon_i$ follows a Cauchy distribution.
	\item After fitting the model, the Assumption (A2), $\epsilon\perp A|\Xbf$, can be verified by applying conditional independence test with $\hat{r}(\hat{\betabf},\hat{\gammabf})$ and $A$ given $\Xbf$, where $\hat{r}(\hat{\betabf},\hat{\gammabf})$ is the estimated residual and $\hat{r}(\hat{\betabf},\hat{\gammabf})=Y-\varphi(\Xbf;\hat{\gammabf})-\{A-\pi(\Xbf)\}C(\Xbf;\hat{\betabf})$. See \cite{lawrance1976conditional,su2007consistent,song2007testing,huang2010testing,zhang2012kernel} for more discussion of conditional independence hypothesis tests. In particular, we demonstrate the usefulness of the test by applying the Kernel-based conditional independence test (KCI-test, \cite{zhang2012kernel}) in Section 5. KCI-test doesn't assume functional forms among variables and thus suits our need.
\end{enumerate}

When the conditional independence assumption ($\epsilon\perp A|\Xbf$) does not hold, $\hat{\betabf}^{R}_{\rho(\tau)}$ may no longer be a consistent estimator of $\betabf_0$. This is intuitively reasonable as $\epsilon$ contains extra information with respect to $A$. In fact, a general result which can be derived in this case is that, $(\hat{\betabf}^{R}_{\rho(\tau)},\hat{\gammabf}^{R}_{\rho(\tau)})$ minimizes a weighed mean-square error loss function with specification error \citep*{angrist2006quantile,lee2013interpretation}.

Instead of assuming response $Y$ takes an additive error term $\epsilon$ as in \eqref{eq:model_nointeraction}, we assume the conditional quantile function $Q_\tau(\Xbf,A)=\varphi_0(\Xbf)+\{A-\pi(\Xbf)\}C(\Xbf;\betabf_0(\tau))$, where we redundantly represent the baseline function and contrast function as $\varphi_0(\cdot)$ and $C(\cdot)$ respectively. Notice that we use $\betabf_0(\tau)$ instead of $\betabf_0$ to emphasize that the true $\betabf$ may vary with respect to $\tau$. The proposed model is $\hat{Q}(\betabf,\gammabf)=\varphi(\Xbf;\gammabf)+\{A-\pi(\Xbf)\}C(\Xbf;\betabf)$ with $C(\Xbf;\betabf)$ correctly specified. Define
\begin{equation}
\label{eq:population_beta_tau}
\left(\betabf(\tau),\gammabf(\tau)\right)=\underset{\betabf\in\mathcal{B},\gammabf\in\Gamma}\argmin\expect
\left[\rho_{\tau}\{Y-\hat{Q}(\betabf,\gammabf)\}
-\rho_{\tau}\{Y-\hat{Q}(\betabf',\gammabf')\}\right]
\end{equation}
where $(\betabf',\gammabf')$ is any fixed point in $\mathcal{B}\times\Gamma$. Define the QR specification error as $\Delta_{\tau}(\Xbf,A;\betabf,\gammabf)\triangleq\hat{Q}(\betabf,\gammabf)-Q_\tau(\Xbf,A)$. Define the quantile-specific residual
as $\epsilon_{\tau}\triangleq Y-Q_\tau(\Xbf,A)$ with conditional density function $f_{\epsilon_{\tau}}(\cdot|\Xbf,A)$. Then we have the following approximation theorem. The proof of the theorem follows Theorem 1 of \cite{angrist2006quantile}, and is omitted for brevity.
\begin{theorem} Suppose that (i) the conditional density $f_{Y}(y|\Xbf,A)$ exists a.s.; (ii)$\expect[Q_{\tau}(\Xbf,A)]$ and $\expect[\Delta^2_{\tau}(\Xbf,A;\betabf,\gammabf)]$ are finite; (iii) $\left(\betabf(\tau),\gammabf(\tau)\right)$ uniquely solves \eqref{eq:population_beta_tau}. Then
	\begin{equation}
	\left(\betabf(\tau),\gammabf(\tau)\right)=\underset{\betabf, \gammabf}\argmin\expect[w_{\tau}(\Xbf,A;\betabf,\gammabf)\Delta^2_{\tau}(\Xbf,A;\betabf,\gammabf)]
	\end{equation}
	where
	\begin{equation}
	w_{\tau}(\Xbf,A;\betabf,\gammabf)=\int_{0}^{1}(1-u)f_{\epsilon_{\tau}}(u\Delta_{\tau}(\Xbf,A;\betabf,\gammabf)|\Xbf,A)du.
	\end{equation}
	\label{thm:approximation_quantile}
\end{theorem}
\noindent \textbf{Remarks:}
\begin{enumerate}
	\item Theorem 2 shows that $\hat{Q}\left(\betabf(\tau),\gammabf(\tau)\right)$ is a weighted least square approximation to $Q_{\tau}(\Xbf,A)$. In other word, $\varphi(\Xbf;\gammabf(\tau))+\{A-\pi(\Xbf)\}C(\Xbf;\betabf(\tau))$ is close to $\varphi_0(\Xbf)+\{A-\pi(\Xbf)\}C(\Xbf;\betabf_0(\tau))$. So even though it is not true that $\betabf(\tau)=\betabf_0(\tau)$ holds exactly, the difference between them is small in general . This coupled with the fact that $\hat{\betabf}^{R}_{\rho(\tau)}\inprob\betabf(\tau)$ (proved in Theorem \ref{thm:asymptotic_normality_pinball}), leads to the conclusion that approximately ITR $\hat{g}^{R}_{\rho(\tau)}(\xbf)$ $(\triangleq\identity\{C(\xbf;\hat{\betabf}^R_{\rho(\tau)})>0\})$ maximizes the $\tau$-th conditional quantile. This observation is justified numerically in Section 4.2.
	\item When there exists $\gammabf_0\in\Gamma$ such that $\varphi_0(\Xbf)\equiv\varphi(\Xbf;\gammabf_0)$, then we have $\betabf(\tau)=\betabf_0(\tau)$.
\end{enumerate}

\subsection{Consistency of Robust Regression: Other Losses}
Under model \eqref{eq:model_nointeraction} and the assumption $\epsilon\perp A|\Xbf$, similar consistency results can be established for Huber loss and the $\epsilon$-insensitive loss, as stated
in Theorem \ref{thm:consistency_huber}.
\begin{theorem}
	\label{thm:consistency_huber}
	Under regularity conditions (A1)-(A8), if the contrast function in \eqref{eq:model_nointeraction} is correctly specified and $\pi(\xbf)$ is known, then we have
	\begin{enumerate}[(a)]
		\item $\hat{\betabf}^{R}_{H(\alpha)}\inprob\betabf_0$ for all $\alpha>0$, where $\hat{\betabf}^{R}_{H(\alpha)}$ is the solution of \eqref{eq:A-general-loss} when $M(x)=H_{\alpha}(x)$;
		\item $\hat{\betabf}^{R}_{J(\epsilon)}\inprob\betabf_0$ for all $\epsilon>0$, where $\hat{\betabf}^{R}_{J(\epsilon)}$ is the solution of \eqref{eq:A-general-loss} when $M(x)=J_{\epsilon}(x)$.
	\end{enumerate}
\end{theorem}
\subsection{Asymptotic Normality: Pinball Loss}
Without loss of generality, in this section we assume both the $\varphi(\Xbf;\gammabf)$ and $C(\Xbf;\betabf)$ take the linear form: $\varphi(\Xbf;\gammabf)=\tilde{\Xbf}\trans\gammabf$ and $C(\Xbf;\betabf)=\tilde{\Xbf}\trans\betabf$, where $\tilde{\Xbf}=(1,\Xbf\trans)\trans$. Denote $\hat{\betabf}(\tau)=\hat{\betabf}^{R}_{\rho(\tau)}$ and $\hat{\gammabf}(\tau)=\hat{\gammabf}^{R}_{\rho(\tau)}$. Denote $\Wbf=(\{A-\pi(\Xbf)\}\tilde{\Xbf}\trans,\tilde{\Xbf}\trans)\trans$, $\thetabf(\tau)=(\betabf(\tau)\trans,\gammabf(\tau)\trans)\trans$, $\hat{\thetabf}(\tau)=(\hat{\betabf}(\tau)\trans,\hat{\gammabf}(\tau)\trans)\trans$ and $J(\tau)\triangleq\expect\left[f_Y(\Wbf\trans\thetabf(\tau)|\Xbf,A)\Wbf\Wbf\trans\right]$. Under the following regularity conditions, which is the same as the assumptions assumed in \cite{angrist2006quantile} and \cite{lee2013interpretation}, we have the asymptotic normality of $\hat{\thetabf}(\tau)$, which is given in Theorem \ref{thm:asymptotic_normality_pinball}.
\begin{itemize}
	\item[(B1)] $\{(Y_i,\Xbf_i,A_i,\epsilon_i),i=1,\ldots,n\}$ are i.i.d random variables;
	\item[(B2)] the conditional density $f_Y(y|\Xbf=\xbf,A=a))$ exists, and is bounded and uniformly continuous in y, uniformly in $\xbf$ over the support of $\Xbf$;
	\item[(B3)] $J(\tau)$ is positive definite for all $\tau\in(0,1)$, where $\thetabf(\tau)$ is uniquely defined in \eqref{eq:population_beta_tau};
	\item[(B4)] $\expect\|\Xbf\|^{2+\epsilon}$ for some $\epsilon>0$.
\end{itemize}
\begin{theorem}
	\label{thm:asymptotic_normality_pinball}
	If regularity conditions (B1)-(B4) are hold, we have
	\begin{enumerate}
		\item (\textbf{Uniform Consistency}) $\sup_{\tau}\|\hat{\thetabf}(\tau)-\thetabf(\tau)\|=o_p(1)$;
		\item (\textbf{Asymptotic Normality}) $J(\cdot)\sqrt{n}(\hat{\thetabf}(\cdot)-\thetabf(\cdot))$ converge in distribution to a zero mean Gaussian process with covariance function $\Sigma(\tau,\tau')$ defined as
		\begin{equation}
		\Sigma(\tau,\tau')=\expect\left[\left(\tau-\identity\{Y<\Wbf\trans\thetabf(\tau)\}\right)
		\left(\tau'-\identity\{Y<\Wbf\trans\thetabf(\tau)\}\right)\Wbf\Wbf\trans\right].
		\end{equation}
	\end{enumerate}
\end{theorem}
The proof is given in \cite{angrist2006quantile}, and the asymptotic covariance matrix of $\hat{\thetabf}(\tau)$ can be estimated by either a bootstrap procedure \citep{hahn1997bayesian} or a nonparametric kernel method \citep{angrist2006quantile}. We adopt the parametric bootstrap approach to estimate the asymptotic covariance matrix in Section 5. Under model \eqref{eq:model_nointeraction} the result of Theorem \ref{thm:asymptotic_normality_pinball} can be further simplified, which is given in Theorem 5.
\begin{theorem}
	Under the condition of Theorem 4, if further we assume $Y=\varphi_0(\Xbf)+\{A-\pi(\Xbf)\}\tilde{\Xbf}\trans\betabf_0+\epsilon$, and $\epsilon\perp A|\Xbf$, then
	\begin{enumerate}
		\item $\sup_{\tau}\|\hat{\betabf}(\tau)-\betabf_0\|=o_p(1)$;
		\item $\sqrt{n}(\hat{\betabf}(\tau)-\betabf_0)\indist N(\zerobf, J_{11}^{-1}(\tau)\Sigma_{11}(\tau,\tau)J_{11}^{-1}(\tau))$, where
		\begin{align*}
		J_{11}(\tau)=& \expect\left[f_{\epsilon}\left(\tilde{\Xbf}\trans\gammabf(\tau)-\varphi_0(\Xbf)|\Xbf\right)
		\pi(\Xbf)\{1-\pi(\Xbf)\}\tilde{\Xbf}\tilde{\Xbf}\trans\right],\\
		\Sigma_{11}(\tau,\tau)=& \expect\left\{\left[\tau-\identity\{\epsilon<\tilde{\Xbf}\trans\gammabf(\tau)-\varphi_0(\Xbf)\}\right]^2
		\pi(\Xbf)\{1-\pi(\Xbf)\}\tilde{\Xbf}\tilde{\Xbf}\trans\right\}.
		\end{align*}
		Furthermore, we have $\Sigma_{11}(\tau,\tau)\leq\left(\tau^2+|1-2\tau|\right)\expect\left[\pi(\Xbf)\{1-\pi(\Xbf)\}\tilde{\Xbf}\tilde{\Xbf}\trans\right]$.
	\end{enumerate}
\end{theorem}


Comparing the asymptotic normality of $\hat{\betabf}(\tau)$ with $\hat{\betabf}^A_{LS}$ yields interesting insights. Assuming that $\expect(Y|\Xbf,A)=\varphi_0(\Xbf)+\{A-\pi(\Xbf)\}\tilde{\Xbf}\trans\betabf_0$ and $(\betabf_0,\gammabf^*)=\argmin_{(\betabf,\gammabf)}\expect[Y-\varphi(\Xbf;\gammabf)-\{A-\pi(\Xbf)\}\tilde{\Xbf}\trans\betabf]^2$, the asymptotic normality property of $\hat{\betabf}^{A}_{LS}$ can then be established, which is summarized in Theorem \ref{thm:lsA}. Its proof has been omitted, and readers are referred to \cite{lu2011variable}.
\begin{theorem}
	\label{thm:lsA}
	Under the regularity condition of A1-A4 of \cite{lu2011variable},
	\begin{equation}
	\sqrt{n}(\hat{\betabf}^{A}_{LS}-\betabf_0)\indist N(0,U_{11}^{-1}\Omega_{11}U_{11}^{-1}),
	\end{equation}
	where $U_{11}=\expect\left[\pi(\Xbf)\{1-\pi(\Xbf)\}\tilde{\Xbf}\tilde{\Xbf}\trans\right]$ and 
	$$\Omega_{11}=\expect\left[\left\{\varphi_0(\Xbf)-\varphi(\Xbf;\gammabf^*)+\epsilon\right\}^2
	\pi(\Xbf)\{1-\pi(\Xbf)\}\tilde{\Xbf}\tilde{\Xbf}\trans\right]$$
\end{theorem}
\noindent \textbf{Remarks:}
\begin{enumerate}
	\item When the family of functions $\{\varphi(\Xbf;\gammabf),\gammabf\in\Gamma\}$ cannot well approximate the unknown baseline function $\varphi_0(\Xbf)$, the $\Omega_{11}$ term in the asymptotic variance of $\hat{\betabf}^{A}_{LS}$ may explode, which makes $\hat{\betabf}^{A}_{LS}$ less efficient than $\hat{\betabf}(\tau)$.
	\item When $Y=\tilde{\Xbf}\trans\gamma_0+\{A-\pi(\Xbf)\}\tilde{\Xbf}\trans\betabf_0+\epsilon$, $\epsilon\perp(A,\Xbf)$, $\pi(\Xbf)\equiv0.5$ and $\epsilon\sim N(0,\sigma^2)$, the asymptotic variance of $\hat{\betabf}(\tau=0.5)$ is $2\pi\sigma^2\expect(\tilde{\Xbf}\tilde{\Xbf}\trans)^{-1}$, which is strictly larger than $4\sigma^2\expect(\tilde{\Xbf}\tilde{\Xbf}\trans)^{-1}$ (the asymptotic variance of $\hat{\betabf}^{A}_{LS}$).
\end{enumerate}

\section{Numerical Results: Simulation Studies}
\label{section:simulation}
To demonstrate finite sample performance of the proposed robust regression methods for optimal treatment rule estimation, we conduct two simulation studies: the errors independent with treatments, and the errors interactive with treatments, respectively.
\subsection{Simulation Study I:  error terms independent with treatment}
We consider the following two models with p=3,
\begin{itemize}
	\item Model I:
	\begin{equation*}
	Y_i=1+(X_{i1}-X_{i2})(X_{i1}+X_{i3})+\{A_i-\pi(\Xbf_{i})\}\betabf_0\trans\tilde{\Xbf}_i+\sigma(\Xbf_{i})\epsilon_i,
	\end{equation*}
	where $\Xbf_{i}=(X_{i1},X_{i2},X_{i3})\trans$ are multivariate normal with mean 0, variance 1, and $\mathrm{Corr}(X_{ij},X_{ik})=0.5^{|j-k|}$, $\tilde{\Xbf}_i=(1,\Xbf_i\trans)\trans$ and $\betabf_0=(0,1,-1,1)\trans$.
	\item Model II:
	\begin{equation*}
	Y_i=\gammabf_0\trans\tilde{\Xbf}_i+\{A_i-\pi(\Xbf_{i})\}\betabf_0\trans\tilde{\Xbf}_i+\sigma(\Xbf_{i})\epsilon_i,
	\end{equation*}
	where $\gammabf_0\trans=(0.5,4,1,-3)$, and $\Xbf_{i}$, $\tilde{\Xbf}_{i}$ and $\betabf_0$ are the same as Model I.
\end{itemize}

We take linear forms for both the baseline and the contrast functions, where $\varphi(\Xbf;\gammabf)=\gammabf\trans\tilde{\Xbf}$ and $C(\Xbf;\betabf)=\betabf\trans\tilde{\Xbf}$. We assume the propensity scores $\pi(\cdot)$ are known, and we study both the constant case $(\pi(\Xbf_i)=0.5)$ and the non-constant case $(\pi(\Xbf_i)=\mathrm{logit}(\Xbf_{i1}-\Xbf_{i2}))$. In addition, We consider two different $\sigma(\Xbf_{i})$ functions, i.e., the homogeneous case with $\sigma(\Xbf_{i})=1$, and the heterogenous case with $\sigma(\Xbf_{i})=0.5+(X_{i1}-X_{i2})^2$. The simulation results under constant and non-constant propensity scores are similar. Thus, for brevity, we only report the constant case and allocate the result of non-constant case to the Appendix B. The results of Model I and II with constant propensity score are given in Table \ref{table:modelI_constant} and \ref{table:modelII_constant} respectively.

\begin{table}[ht]
	\setlength{\tabcolsep}{4pt}
	\caption{Summary result of Model I with constant propensity scores. LS stands for lsA-learning. P(0.5) stands for robust regression with pinball loss and parameter $\tau=0.5$. P(0.25) stands for robust regression with pinball loss and parameter $\tau=0.25$. Huber stands for robust regression with Huber loss, where parameter $\alpha$ is tuned automatically with R function rlm. Column $\delta_{0.5}$ is multiplied by 10.}
	\label{table:modelI_constant}
	\begin{center}
		\begin{tabular}{@{}ll ccc ccc ccc}
			\toprule
			\multicolumn{11}{c}{{\textbf{Homogeneous Error}}}\\
			\hline
			& & \multicolumn{3}{c}{{{\bfseries Normal}}}
			& \multicolumn{3}{c}{{{\bfseries Log-Normal}}}
			& \multicolumn{3}{c}{{{\bfseries Cauchy}}}\\
			\cmidrule(r){3-5}
			\cmidrule(lr){6-8}
			\cmidrule(l){9-11}
			n & method & mse  & PCD & $\delta_{0.5}$ & mse  & PCD  & $\delta_{0.5}$ & mse  & PCD & $\delta_{0.5}$\\
			\cmidrule(r){3-5}
			\cmidrule(lr){6-8}
			\cmidrule(l){9-11}
			100 & LS & 1.32 (0.040) & 80.7 & 1.06 & 2.36 (0.081) & 75.7 & 1.57 &  & 58.4 & 3.75 \\
			& P(0.5) & 1.44 (0.042) & 80.1 & 1.13 & 1.73 (0.051) & 78.0 & 1.31 & 2.69 (0.077) & 75.2 & 1.63 \\
			& P(0.25) & 1.90 (0.057) & 78.3 & 1.34 & 1.63 (0.051) & 79.0 & 1.29 & 5.29 (0.168) & 70.4 & 2.25 \\
			& Huber & 1.15 (0.034) & 81.9 & 0.93 & 1.45 (0.044) & 79.9 & 1.13 & 2.61 (0.072) & 74.9 & 1.66 \\
			200 & LS & 0.68 (0.021) & 85.6 & 0.59 & 1.10 (0.033) & 82.0 & 0.91 &  & 58.7 & 3.70 \\
			& P(0.5) & 0.73 (0.021) & 85.3 & 0.62 & 0.78 (0.021) & 84.1 & 0.70 & 1.23 (0.037) & 81.3 & 0.99 \\
			& P(0.25) & 0.92 (0.028) & 84.0 & 0.75 & 0.70 (0.023) & 86.0 & 0.59 & 2.48 (0.079) & 75.7 & 1.64 \\
			& Huber & 0.58 (0.017) & 86.8 & 0.50 & 0.66 (0.018) & 85.5 & 0.58 & 1.24 (0.035) & 80.8 & 1.03 \\
			400 & LS & 0.33 (0.009) & 90.3 & 0.26 & 0.56 (0.016) & 87.1 & 0.46 &  & 59.2 & 3.61 \\
			& P(0.5) & 0.35 (0.010) & 90.0 & 0.29 & 0.37 (0.010) & 89.0 & 0.34 & 0.56 (0.016) & 87.1 & 0.48 \\
			& P(0.25) & 0.43 (0.013) & 89.1 & 0.34 & 0.33 (0.010) & 90.7 & 0.25 & 1.16 (0.037) & 82.9 & 0.86 \\
			& Huber & 0.28 (0.008) & 91.1 & 0.22 & 0.31 (0.009) & 90.2 & 0.27 & 0.58 (0.017) & 86.7 & 0.49 \\
			800 & LS & 0.17 (0.005) & 93.2 & 0.13 & 0.26 (0.008) & 90.9 & 0.23 &  & 59.4 & 3.59 \\
			& P(0.5) & 0.17 (0.005) & 93.1 & 0.13 & 0.19 (0.005) & 92.1 & 0.17 & 0.29 (0.009) & 90.7 & 0.24 \\
			& P(0.25) & 0.22 (0.007) & 92.4 & 0.16 & 0.18 (0.006) & 93.6 & 0.12 & 0.59 (0.019) & 87.3 & 0.48 \\
			& Huber & 0.14 (0.004) & 93.8 & 0.11 & 0.16 (0.005) & 93.1 & 0.14 & 0.29 (0.008) & 90.5 & 0.25 \\
			\hline
			\multicolumn{11}{c}{{\textbf{Heterogenous Error}}}\\
			\hline
			& & \multicolumn{3}{c}{{{\bfseries Normal}}}
			& \multicolumn{3}{c}{{{\bfseries Log-Normal}}}
			& \multicolumn{3}{c}{{{\bfseries Cauchy}}}\\
			\cmidrule(r){3-5}
			\cmidrule(lr){6-8}
			\cmidrule(l){9-11}
			n & method & mse  & PCD & $\delta_{0.5}$ & mse  & PCD  & $\delta_{0.5}$ & mse  & PCD & $\delta_{0.5}$\\
			\cmidrule(r){3-5}
			\cmidrule(lr){6-8}
			\cmidrule(l){9-11}
			100 & LS & 3.24 (0.110) & 74.7 & 1.70 & 8.98 (0.561) & 68.6 & 2.44 &  & 56.2 & 4.05 \\
			& P(0.5) & 1.70 (0.060) & 80.5 & 1.08 & 1.80 (0.064) & 80.1 & 1.08 & 3.45 (0.124) & 75.1 & 1.69 \\
			& P(0.25) & 2.50 (0.085) & 77.4 & 1.42 & 2.51 (0.079) & 76.8 & 1.46 & 9.13 (0.341) & 67.2 & 2.66 \\
			& Huber & 1.70 (0.057) & 80.4 & 1.10 & 1.87 (0.063) & 79.2 & 1.16 & 4.27 (0.155) & 72.8 & 1.93 \\
			200 & LS & 1.54 (0.050) & 80.6 & 1.06 & 4.71 (0.244) & 73.4 & 1.85 &  & 55.2 & 4.17 \\
			& P(0.5) & 0.78 (0.028) & 86.7 & 0.53 & 0.90 (0.032) & 85.3 & 0.63 & 1.49 (0.052) & 81.9 & 0.95 \\
			& P(0.25) & 1.16 (0.039) & 83.5 & 0.81 & 1.23 (0.039) & 82.0 & 0.91 & 3.95 (0.150) & 73.2 & 1.90 \\
			& Huber & 0.77 (0.025) & 86.4 & 0.55 & 0.94 (0.032) & 84.5 & 0.69 & 1.94 (0.071) & 79.3 & 1.19 \\
			400 & LS & 0.80 (0.026) & 86.0 & 0.58 & 2.69 (0.136) & 77.8 & 1.34 &  & 54.7 & 4.26 \\
			& P(0.5) & 0.39 (0.013) & 90.5 & 0.27 & 0.44 (0.017) & 89.6 & 0.32 & 0.71 (0.024) & 86.9 & 0.50 \\
			& P(0.25) & 0.56 (0.019) & 88.8 & 0.37 & 0.66 (0.020) & 86.9 & 0.50 & 1.70 (0.055) & 79.6 & 1.17 \\
			& Huber & 0.38 (0.012) & 90.4 & 0.27 & 0.48 (0.017) & 88.8 & 0.36 & 0.91 (0.029) & 84.9 & 0.65 \\
			800 & LS & 0.41 (0.013) & 89.9 & 0.29 & 1.35 (0.150) & 83.1 & 0.82 &  & 56.5 & 4.00 \\
			& P(0.5) & 0.18 (0.006) & 93.6 & 0.12 & 0.20 (0.007) & 92.6 & 0.16 & 0.36 (0.013) & 91.0 & 0.25 \\
			& P(0.25) & 0.28 (0.009) & 92.2 & 0.18 & 0.31 (0.010) & 90.8 & 0.24 & 0.89 (0.031) & 85.8 & 0.60 \\
			& Huber & 0.19 (0.006) & 93.3 & 0.13 & 0.22 (0.007) & 92.1 & 0.18 & 0.47 (0.017) & 89.2 & 0.34 \\
			\bottomrule
		\end{tabular}
	\end{center}
\end{table}

\begin{table}[ht]
	\setlength{\tabcolsep}{4pt}
	\caption{Summary result of Model II with constant propensity scores. LS stands for lsA-learning. P(0.5) stands for robust regression with pinball loss and parameter $\tau=0.5$. P(0.25) stands for robust regression with pinball loss and parameter $\tau=0.25$. Huber stands for robust regression with Huber loss, where parameter $\alpha$ is tuned automatically with R function rlm. Column $\delta_{0.5}$ is multiplied by 10.}
	\label{table:modelII_constant}
	\begin{center}
		\begin{tabular}{@{}ll ccc ccc ccc}
			\toprule
			\multicolumn{11}{c}{{\textbf{Homogeneous Error}}}\\
			\hline
			& & \multicolumn{3}{c}{{{\bfseries Normal}}}
			& \multicolumn{3}{c}{{{\bfseries Log-Normal}}}
			& \multicolumn{3}{c}{{{\bfseries Cauchy}}}\\
			\cmidrule(r){3-5}
			\cmidrule(lr){6-8}
			\cmidrule(l){9-11}
			n & method & mse  & PCD & $\delta_{0.5}$ & mse  & PCD  & $\delta_{0.5}$ & mse  & PCD & $\delta_{0.5}$\\
			\cmidrule(r){3-5}
			\cmidrule(lr){6-8}
			\cmidrule(l){9-11}
			100 & LS & 0.24 (0.006) & 91.1 & 0.21 & 1.23 (0.061) & 82.4 & 0.87 &  & 58.6 & 3.73 \\
			& P(0.5) & 0.36 (0.010) & 89.0 & 0.32 & 0.39 (0.012) & 88.8 & 0.34 & 0.80 (0.024) & 84.2 & 0.69 \\
			& P(0.25) & 0.45 (0.012) & 87.8 & 0.40 & 0.13 (0.004) & 93.4 & 0.12 & 2.37 (0.083) & 76.0 & 1.49 \\
			& Huber & 0.25 (0.007) & 90.8 & 0.22 & 0.31 (0.010) & 90.3 & 0.26 & 0.99 (0.029) & 82.4 & 0.84 \\
			200 & LS & 0.11 (0.003) & 93.7 & 0.10 & 0.52 (0.018) & 87.3 & 0.45 &  & 58.7 & 3.69 \\
			& P(0.5) & 0.17 (0.005) & 92.4 & 0.16 & 0.17 (0.005) & 92.4 & 0.15 & 0.32 (0.009) & 89.5 & 0.30 \\
			& P(0.25) & 0.20 (0.005) & 91.8 & 0.18 & 0.06 (0.002) & 95.6 & 0.05 & 1.03 (0.033) & 82.1 & 0.88 \\
			& Huber & 0.12 (0.003) & 93.6 & 0.11 & 0.13 (0.003) & 93.5 & 0.12 & 0.43 (0.013) & 87.9 & 0.40 \\
			400 & LS & 0.05 (0.001) & 95.7 & 0.05 & 0.26 (0.008) & 90.7 & 0.23 &  & 59.4 & 3.60 \\
			& P(0.5) & 0.09 (0.002) & 94.5 & 0.08 & 0.09 (0.002) & 94.5 & 0.08 & 0.15 (0.004) & 92.8 & 0.14 \\
			& P(0.25) & 0.10 (0.002) & 94.2 & 0.09 & 0.03 (0.001) & 96.9 & 0.02 & 0.44 (0.012) & 87.9 & 0.39 \\
			& Huber & 0.06 (0.001) & 95.5 & 0.05 & 0.06 (0.002) & 95.4 & 0.06 & 0.21 (0.006) & 91.6 & 0.19 \\
			800 & LS & 0.03 (0.001) & 96.9 & 0.03 & 0.13 (0.004) & 93.5 & 0.11 &  & 59.4 & 3.58 \\
			& P(0.5) & 0.04 (0.001) & 96.1 & 0.04 & 0.04 (0.001) & 96.2 & 0.04 & 0.07 (0.002) & 95.1 & 0.06 \\
			& P(0.25) & 0.05 (0.001) & 95.8 & 0.05 & 0.01 (0.000) & 97.9 & 0.01 & 0.20 (0.005) & 91.5 & 0.19 \\
			& Huber & 0.03 (0.001) & 96.8 & 0.03 & 0.03 (0.001) & 96.8 & 0.03 & 0.10 (0.002) & 94.2 & 0.09 \\
			\hline
			\multicolumn{11}{c}{{\textbf{Heterogenous Error}}}\\
			\hline
			& & \multicolumn{3}{c}{{{\bfseries Normal}}}
			& \multicolumn{3}{c}{{{\bfseries Log-Normal}}}
			& \multicolumn{3}{c}{{{\bfseries Cauchy}}}\\
			\cmidrule(r){3-5}
			\cmidrule(lr){6-8}
			\cmidrule(l){9-11}
			n & method & mse  & PCD & $\delta_{0.5}$ & mse  & PCD  & $\delta_{0.5}$ & mse  & PCD & $\delta_{0.5}$\\
			\cmidrule(r){3-5}
			\cmidrule(lr){6-8}
			\cmidrule(l){9-11}
			100 & LS & 1.97 (0.072) & 79.8 & 1.13 & 7.75 (0.514) & 70.4 & 2.22 &  & 56.4 & 4.02 \\
			& P(0.5) & 0.84 (0.029) & 86.1 & 0.55 & 1.21 (0.045) & 84.3 & 0.74 & 1.82 (0.071) & 80.5 & 1.07 \\
			& P(0.25) & 1.37 (0.049) & 82.1 & 0.90 & 1.56 (0.051) & 80.5 & 1.04 & 6.20 (0.261) & 69.8 & 2.25 \\
			& Huber & 0.84 (0.031) & 85.9 & 0.57 & 1.33 (0.046) & 82.8 & 0.85 & 2.69 (0.106) & 77.0 & 1.42 \\
			200 & LS & 0.99 (0.035) & 84.7 & 0.66 & 4.16 (0.237) & 75.2 & 1.62 &  & 55.1 & 4.19 \\
			& P(0.5) & 0.41 (0.014) & 90.2 & 0.28 & 0.58 (0.024) & 89.4 & 0.37 & 0.79 (0.030) & 86.7 & 0.52 \\
			& P(0.25) & 0.64 (0.021) & 87.4 & 0.45 & 0.74 (0.024) & 86.1 & 0.54 & 2.48 (0.096) & 76.9 & 1.40 \\
			& Huber & 0.39 (0.013) & 90.3 & 0.27 & 0.69 (0.027) & 87.7 & 0.45 & 1.17 (0.044) & 83.4 & 0.78 \\
			400 & LS & 0.51 (0.018) & 89.0 & 0.35 & 2.48 (0.133) & 79.3 & 1.20 &  & 54.7 & 4.25 \\
			& P(0.5) & 0.20 (0.007) & 93.2 & 0.14 & 0.29 (0.011) & 92.6 & 0.17 & 0.32 (0.011) & 91.2 & 0.22 \\
			& P(0.25) & 0.30 (0.009) & 91.3 & 0.22 & 0.39 (0.012) & 89.9 & 0.28 & 0.99 (0.030) & 83.0 & 0.78 \\
			& Huber & 0.20 (0.007) & 93.2 & 0.14 & 0.34 (0.012) & 91.4 & 0.22 & 0.53 (0.016) & 88.4 & 0.37 \\
			800 & LS & 0.25 (0.008) & 92.2 & 0.17 & 1.25 (0.159) & 84.2 & 0.73 &  & 56.4 & 4.00 \\
			& P(0.5) & 0.10 (0.004) & 95.3 & 0.07 & 0.14 (0.006) & 94.7 & 0.09 & 0.16 (0.006) & 93.9 & 0.11 \\
			& P(0.25) & 0.14 (0.005) & 94.0 & 0.10 & 0.18 (0.006) & 92.9 & 0.14 & 0.49 (0.015) & 88.0 & 0.39 \\
			& Huber & 0.09 (0.004) & 95.3 & 0.06 & 0.17 (0.006) & 93.9 & 0.11 & 0.26 (0.009) & 91.8 & 0.19 \\
			\bottomrule
		\end{tabular}
	\end{center}
\end{table}

Comparison is made among four methods. They are: lsA-learning, robust regression with $\rho_{0.5}$ (RR($\rho_{0.5}$)), robust regression with $\rho_{0.25}$ (RR($\rho_{0.25}$)), and robust regression with Huber loss (RR(H)). The error terms $\epsilon_i$ are taken as standard i.i.d. normal, log-normal or Cauchy distribution, and independent with both $A$ and $\Xbf$. It is easy to check that the conditional independence assumption $\epsilon\perp A|\Xbf$ is satisfied, and $g^{\mathrm{opt}}_{\mu}=g^{\mathrm{opt}}_{\tau}=\identity\{\betabf_0\trans\tilde{\Xbf}_i>0\}$. We consider four different sample sizes 100, 200, 400 and 800. To evaluate the performance of each method, we compare three groups of criteria: (1) the mean squared error $\|\hat{\betabf}-\betabf_{0}\|^2_2$ (mse), which measures the distance between estimated parameters and the true parameter $\betabf_0$; (2) the percentage of making correct decisions (PCD), which are calculated based on a validation set with 10000 observations. Specifically, we take the formula $100*\left(1-\sum_{i=1}^{N_T}|\identity\{\hat{\betabf}\trans\tilde{\Xbf}_i>0\}-\identity\{\betabf_0\trans\tilde{\Xbf}_i>0\}|/N_T\right)$ with $N_T=10000$; (3) the differences of $V_\mu(g)$ and $V_{0.5-q}(g)$ between the optimal ITR and the estimated ITR, where $\delta_{\mu}=V_{\mu}(g_{\mu}^{\mathrm{opt}})-V_{\mu}(\hat{g})$ and $\delta_{\tau}=V_{\tau-q}(g_{\mu}^{\mathrm{opt}})-V_{\tau-q}(\hat{g})$, $\forall\tau\in(0,1)$. $V_{\mu}(g)$ and $V_{\tau-q}(g)$ (defined in Section 2.1) are estimated from the validation set as well, and they evaluate the overall performance of an ITR $g$, where the former one focuses on the response's mean and the latter one focuses on the response's conditional $\tau$-th quantile. Under our setting, $\delta_{\mu}=\delta_{0.5}$ when they both exists. Thus, only $\delta_{0.5}$ is reported. For each scenario, we take 1000 replications. All numbers in the tables are based on the sample average of all replications. We further report the standard errors of mse to evaluate the variability of the corresponding statistics.

When the propensity score is constant, lsA-learning is equivalent to both Q- and A-learning under our setting. If we compare the performance of the methods under homogeneous and heterogeneous errors, the first thing we find is that lsA-learning works much worse under the heterogeneous errors, while all other methods are generally less affected by the heterogeneity of the errors. When the baseline function is misspecified as in Model I, under the homogeneous normal errors, RR(H) works slightly better than lsA-learning, while $\mathrm{RR}(\rho_{0.25})$ works the worst. However, the difference in general is small. For the homogeneous log-normal errors, again RR(H) works the best, while $\mathrm{RR}(\rho_{0.5})$ and $\mathrm{RR}(\rho_{0.25})$ have similar performance, and lsA-learning works the worst. Under the homogeneous Cauchy errors, $\mathrm{RR}(\rho_{0.5})$ works the best and RR(H) has a close performance. The lsA-learning is no longer consistent, and its mse explodes. The actual numbers are too large and thus leave as blank in Table \ref{table:modelI_constant} and \ref{table:modelII_constant}. Furthermore, with the Cauchy errors, the PCD of lsA-learning are less than 60\% under all scenarios, while other methods' PCD can be as high as 90\%. When baseline function is correctly specified as in Model II, under homogeneous normal errors, lsA-learning performs the best. However, in this case RR(H) also has a very close performance, and thus makes no difference from a practical point of view to choose between these two methods. The results of Model II under other cases draw similar conclusion as Model I. To sum up, the overall conclusion is that, under the conditional independence assumption, the proposed robust regression method RR(M) is more efficient than Q-, A- and lsA-learning in the circumstances when observations have skewed, heterogeneous or heavy-tailed errors. On the other hand, when the error terms indeed follows i.i.d. normal distribution, the loss of efficiency of RR(M) is not significant. This is especially true when Huber loss is applied.

\subsection{Simulation Study II: error terms interactive with treatment}
We consider the following model with p=2,
\begin{equation*}
Y_i=1 + 0.5\sin[\pi(X_{i1}-X_{i2})]+
0.25(1+X_{i1}+2X_{i2})^2+(A_i-\pi(\Xbf_{i}))\thetabf_0\trans\tilde{\Xbf}_i+\sigma(\Xbf_{i},A_i)\epsilon_i,
\end{equation*}
where $\Xbf_{i}=(X_{i1},X_{i2})\trans$, $\tilde{\Xbf}_i=(1,\Xbf_i\trans)\trans$, $\sigma(\Xbf_{i},A_i)=1+A_i d_0 X_{i1}^2$, $\thetabf_0\trans=(0.5,2,-1)$ and $X_{ik}$ are i.i.d. Uniform[-1,1].

Similar as Section 4.1, we take linear forms for both the baseline and the contrast functions, where $\varphi(\Xbf;\gammabf)=\gammabf\trans\tilde{\Xbf}$, $C(\Xbf;\betabf)=\betabf\trans \Wbf$ and $\Wbf=(\tilde{\Xbf},X_{1}^2,X_{2}^2,X_{1}X_{2})$. $d_0=5$, 10 or 15. The error terms $\epsilon_i$ follows i.i.d. N(0,1) or Gamma(1,1)-1 distribution. The propensity scores $\pi(\cdot)$ are known, and we consider  both the constant case $\pi(\Xbf_i)=0.5$ and the non-constant case $\pi(\Xbf_i)=\mathrm{logit}(\Xbf_{i1}-\Xbf_{i2})$. We report only the result of the constant case (Table \ref{table:interacted_constant_ps}), and allocate the non-constant case to Appendix B.

\begin{table}[ht]
	\setlength{\tabcolsep}{2.5pt}
	\caption{Summary results with constant propensity scores when errors interacted with treatment. Least square stands for lsA-learning. Pinball(0.5) stands for robust regression with pinball loss and parameter $\tau=0.5$. Pinball(0.25) stands for robust regression with pinball loss and parameter $\tau=0.25$.  Huber stands for robust regression with Huber loss, where parameter $\alpha$ is tuned automatically with R function rlm.}
	\label{table:interacted_constant_ps}
	\begin{center}
		\begin{tabular}{@{}lll @{\hskip 0.5cm} rrr @{\hskip 0.5cm} rrr @{\hskip 0.5cm} rrr  @{\hskip 0.5cm} rrr}
			\toprule
			& & & \multicolumn{3}{c}{{{\bfseries Least Square}}} 
			& \multicolumn{3}{c}{{{\bfseries Pinball(0.5)}}} 
			& \multicolumn{3}{c}{{{\bfseries Pinball(0.25)}}} 
			& \multicolumn{3}{c}{{{\bfseries Huber}}}\\
			\cmidrule(r){4-6}
			\cmidrule(lr){7-9}
			\cmidrule(lr){10-12}
			\cmidrule(lr){13-15}
			Error & $d_0$ & n & $ \delta_{\mu}$ & $ \delta_{0.5}$ & $ \delta_{0.25}$ & $ \delta_{\mu}$ & $ \delta_{0.5}$ & $ \delta_{0.25}$ & $ \delta_{\mu}$ & $ \delta_{0.5}$ & $ \delta_{0.25}$ 
			& $ \delta_{\mu}$ & $ \delta_{0.5}$ & $ \delta_{0.25}$ \\
			\midrule
			Normal & 5 & 100 & 0.16 & 0.16 & 0.31 & 0.16 & 0.16 & 0.27 & 0.25 & 0.25 & 0.17 & 0.14 & 0.14 & 0.26 \\ 
			&  & 200 & 0.09 & 0.09 & 0.24 & 0.10 & 0.10 & 0.19 & 0.18 & 0.18 & 0.09 & 0.08 & 0.08 & 0.19 \\ 
			&  & 400 & 0.05 & 0.05 & 0.18 & 0.07 & 0.07 & 0.12 & 0.15 & 0.15 & 0.05 & 0.05 & 0.05 & 0.13 \\ 
			&  & 800 & 0.02 & 0.02 & 0.14 & 0.05 & 0.05 & 0.09 & 0.14 & 0.14 & 0.04 & 0.03 & 0.03 & 0.09 \\ 
			& 10 & 100 & 0.28 & 0.28 & 0.92 & 0.22 & 0.22 & 0.81 & 0.39 & 0.39 & 0.40 & 0.21 & 0.21 & 0.82 \\ 
			&  & 200 & 0.19 & 0.19 & 0.85 & 0.15 & 0.15 & 0.71 & 0.33 & 0.33 & 0.28 & 0.13 & 0.13 & 0.72 \\ 
			&  & 400 & 0.12 & 0.12 & 0.79 & 0.10 & 0.10 & 0.60 & 0.30 & 0.30 & 0.23 & 0.09 & 0.09 & 0.63 \\ 
			&  & 800 & 0.06 & 0.06 & 0.73 & 0.07 & 0.07 & 0.50 & 0.27 & 0.27 & 0.22 & 0.06 & 0.06 & 0.54 \\ 
			& 15 & 100 & 0.35 & 0.35 & 1.55 & 0.25 & 0.25 & 1.40 & 0.47 & 0.47 & 0.62 & 0.26 & 0.26 & 1.43 \\ 
			&  & 200 & 0.27 & 0.27 & 1.48 & 0.18 & 0.18 & 1.31 & 0.45 & 0.45 & 0.45 & 0.18 & 0.18 & 1.34 \\ 
			&  & 400 & 0.19 & 0.19 & 1.47 & 0.13 & 0.13 & 1.17 & 0.44 & 0.44 & 0.37 & 0.12 & 0.12 & 1.23 \\ 
			&  & 800 & 0.12 & 0.12 & 1.39 & 0.09 & 0.09 & 1.03 & 0.41 & 0.41 & 0.35 & 0.08 & 0.08 & 1.07 \\ 
			Gamma & 5 & 100 & 0.15 & 0.18 & 0.31 & 0.15 & 0.11 & 0.16 & 0.22 & 0.12 & 0.09 & 0.12 & 0.09 & 0.15 \\ 
			&  & 200 & 0.09 & 0.12 & 0.26 & 0.10 & 0.06 & 0.10 & 0.18 & 0.08 & 0.05 & 0.08 & 0.05 & 0.09 \\ 
			&  & 400 & 0.05 & 0.07 & 0.21 & 0.08 & 0.03 & 0.07 & 0.16 & 0.06 & 0.04 & 0.06 & 0.02 & 0.07 \\ 
			&  & 800 & 0.02 & 0.04 & 0.17 & 0.07 & 0.03 & 0.06 & 0.15 & 0.06 & 0.03 & 0.05 & 0.02 & 0.07 \\ 
			& 10 & 100 & 0.26 & 0.33 & 0.90 & 0.22 & 0.16 & 0.54 & 0.39 & 0.13 & 0.27 & 0.22 & 0.14 & 0.50 \\ 
			&  & 200 & 0.19 & 0.29 & 0.88 & 0.17 & 0.08 & 0.44 & 0.37 & 0.10 & 0.22 & 0.17 & 0.07 & 0.41 \\ 
			&  & 400 & 0.12 & 0.24 & 0.87 & 0.13 & 0.04 & 0.39 & 0.35 & 0.08 & 0.19 & 0.14 & 0.03 & 0.36 \\ 
			&  & 800 & 0.06 & 0.17 & 0.78 & 0.12 & 0.03 & 0.37 & 0.33 & 0.07 & 0.19 & 0.13 & 0.02 & 0.35 \\ 
			& 15 & 100 & 0.36 & 0.57 & 1.52 & 0.30 & 0.31 & 0.98 & 0.53 & 0.19 & 0.40 & 0.32 & 0.28 & 0.89 \\ 
			&  & 200 & 0.28 & 0.53 & 1.51 & 0.22 & 0.19 & 0.81 & 0.55 & 0.16 & 0.29 & 0.24 & 0.16 & 0.71 \\ 
			&  & 400 & 0.19 & 0.47 & 1.50 & 0.17 & 0.13 & 0.73 & 0.57 & 0.15 & 0.26 & 0.21 & 0.11 & 0.63 \\ 
			&  & 800 & 0.11 & 0.43 & 1.50 & 0.15 & 0.11 & 0.71 & 0.58 & 0.15 & 0.24 & 0.18 & 0.09 & 0.62 \\ 
			\bottomrule
		\end{tabular}
	\end{center}
\end{table}

We compare the performance of four methods: lsA-learning, robust regression with $\rho_{0.5}$ ($\mathrm{RR}(\rho_{0.5})$), robust regression with $\rho_{0.25}$ ($\mathrm{RR}(\rho_{0.25})$) and robust regression with Huber loss ($\mathrm{RR}(H)$). We consider four different sample sizes 100, 200, 400 and 800. For each scenario, we again simulate 1000 replications. When error terms are interactive with treatment, the true $\betabf_0$ associated with  $g_{\mu}^{\mathrm{opt}}$ and $g_{\tau}^{\mathrm{opt}}$ are different. Specifically, under our model, $\betabf_0=(\thetabf_0\trans,0,0,0)\trans$ for $g_{\mu}^{\mathrm{opt}}$, $\betabf_0=(\thetabf_0\trans,d_0 F^{-1}_{\epsilon}(0.5),0,0)\trans$ for $g_{0.5}^{\mathrm{opt}}$ and $\betabf_0=(\thetabf_0\trans,d_0 F^{-1}_{\epsilon}(0.25),0,0)\trans$ for $g_{0.25}^{\mathrm{opt}}$. Thus, the two criteria, mse and PCD used in simulation study I,  are no longer meaningful. So we evaluate the performance of methods in this simulation study based on value differences $\delta_{\mu}$, $\delta_{0.5}$ and $\delta_{0.25}$.  


Based on Theorem \ref{thm:lsA}, we can prove that $\hat{g}^{A}_{LS}(\xbf)$ is consistent which converges to $g^{\mathrm{opt}}_{\mu}$ as sample size goes to infinity. This is shown in Table \ref{table:interacted_constant_ps} such that the $\delta_{\mu}$ column for the lsA-learning method converges to 0 as sample size increases. We also know under Normal error terms, $\delta_{0.5}=\delta_{\mu}$. Thus, the $\delta_{0.5}$ column for the lsA-learning method also converges to 0. However, all other columns in Table \ref{table:interacted_constant_ps} converge to a positive constant instead of 0 as sample size goes to infinity.

Another observation we discover from Table \ref{table:interacted_constant_ps} is $\mathrm{RR}(H)$ and $\mathrm{RR}(\rho_{0.5})$ perform similarly. One additional observation we have is even though lsA-learning outperform all other methods in $\delta_{\mu}$ when sample size is large. It may be worse than $\mathrm{RR}(\rho_{0.5})$ and $\mathrm{RR}(H)$ when sample size is small. This is due to the fact that lsA-learning is inefficient under the heteroscedastic or skewed errors. The last observation we have is overall lsA-learning, $\mathrm{RR}(\rho_{0.5})$  and $\mathrm{RR}(\rho_{0.25})$  perform best at the columns $\delta_{\mu}$, $\delta_{0.5}$ and $\delta_{0.25}$ accordingly. The reason is given in the Remark under Theorem \ref{thm:approximation_quantile}, which shows that $\hat{g}^{R}_{\rho(\tau)}$ $(\triangleq\identity\{C(\xbf;\hat{\betabf}^R_{\rho(\tau)})>0\})$ in general approximates the unknown optimal ITR $g^{\mathrm{opt}}_{\tau}$ even when the conditional independence assumption $\epsilon\perp A|\Xbf$ does not hold.    

\section{Application to AIDS study}
\label{section:aids}
We illustrate the proposed robust regression method to data from AIDS Clinical Trials Group Protocol 175 (ACTG175), which has been previously studied by various authors \citep*{leon2003semiparametric,tsiatis2008covariate,zhang2008improving,lu2011variable}. In the study, 2139 HIV-infected subjects were randomized to four different treatment groups in equal proportions, and the treatment groups are zidovudine (ZDV) monotherapy, ZDV + didanosine (ddI), ZDV + zalcitabine, and ddI monotherapy. Following \cite{lu2011variable}, we choose CD4 count $(\mathrm{cells/mm}^3)$ at $20\pm5$ weeks post-baseline as the primary continuous outcome $Y$, and include five continuous covariates and seven binary covariates as our covariates. They are: 1. age (years), 2. weight (kg), 3. karnof=Karnofsky score (scale of 0-100), 4. cd40=CD4 count $(\mathrm{cells/mm}^3)$ at baseline, 5. cd80=CD8 count $(\mathrm{cells/mm}^3)$ at baseline, 6. hemophilia=hemophilia (0=no, 1=yes), 7. homosexuality=homosexual activity (0=no, 1=yes), 8. drugs=history of intravenous drug use (0=no, 1=yes), 9. race (0=white, 1=non-white), 10. gender (0=female, 1=male), 11. str2= antiretroviral history (0=naive, 1=experienced), and 12. sympton=symptomatic status (0=asymptomatic, 1=symptomatic).
For brevity, we only compare the treatment ZDV + didanosine (ddI) $(A=1)$ and ZDV + zalcitabine $(A=0)$, and restrict our samples to subjects receiving these two treatments. Thus, the propensity scores $\pi(\Xbf_i)\equiv 0.5$ in our restricted samples as the patients are assigned into one of two treatments with equal probability.

In our analysis, we assume linear models for both the baseline and the contrast functions. For interpretability, we keep the response $Y$ (the CD4 count) at its original scale, which is also consistent with the way clinicians think about the outcome in practice \citep{tsiatis2008covariate}. We plot the scatter plot of response Y against age. It shows some skewness and heterogeneity. With some preliminary analysis (fitting full model with lsA-learning and RR(M)), we find that only covariates age, homosexuality and race may possibly interact with the treatment. So in our final model, only these three covariates are included in the contrast function, while at the same time we still keep all twelve covariates in the baseline function. The estimated coefficients associated with their corresponding standard errors and p-values are given in Table \ref{table:aids}, where standard errors are estimated with 1000 bootstrap samples (parametric bootstrap) and p-values are calculated with normal approximation. Only coefficients included in the contrast function are shown.

\begin{table}[ht]
	\setlength{\tabcolsep}{2.3pt}
	\caption{Analysis results for AIDS data. Est. stands for estimate; SE stands for standard error; PV stands for p-value. All p-values which are significant at level 0.1 are highlighted.}
	\label{table:aids}
	\begin{center}
		\begin{tabular}{lrrrrrrrrrrrr}
			\toprule
			& \multicolumn{3}{c}{{{\bfseries Least Square}}}
			& \multicolumn{3}{c}{{{\bfseries Pinball(0.5)}}}
			& \multicolumn{3}{c}{{{\bfseries Pinball(0.25)}}}
			& \multicolumn{3}{c}{{{\bfseries Huber}}}\\
			\cmidrule(r){2-4}
			\cmidrule(lr){5-7}
			\cmidrule(l){8-10}
			\cmidrule(l){11-13}
			Variable & Est. & SE & PV & Est. & SE & PV & Est. & SE & PV & Est. & SE & PV \\
			\midrule
			intercept & -42.61 & 32.93 & 0.196 & -33.45 & 37.32 & 0.370 & -35.77 & 39.17 & 0.361 & -42.76 & 31.40 & 0.173 \\ 
			age       & 3.13 & 0.85 & \textbf{0.000} & 2.62 & 0.97 & \textbf{0.007} & 2.46 & 1.06 & \textbf{0.020} & 2.80 & 0.79 & \textbf{0.000} \\ 
			homosexuality      & -40.66 & 16.73 & \textbf{0.015} & -33.18 & 17.68 & 0.061 & -35.38 & 18.28 & 0.053 & -27.33 & 15.19 & 0.072 \\ 
			race      & -25.70 & 17.69 & 0.146 & -33.56 & 18.12 & 0.064 & -34.21 & 18.32 & 0.062 & -25.29 & 16.08 & 0.116 \\
			\bottomrule
		\end{tabular}
	\end{center}
\end{table}

From Tables \ref{table:aids}, we make the following observations. First, lsA-learning (equivalent to Q- and A-learning with this model setting) and robust regression with pinball loss and Huber loss all have estimates with the exact same signs. Second, the estimated coefficients are distinguishable across different methods. Third, the covairiate homosexuality is significant under lsA-learning, but it is not significant under robust regression with either pinball losses or Huber loss, when the significant level $\alpha$ is set to 0.05.

We could further estimate the values $(V_{\mu}(\hat{g}))$ associated with each method by either the inverse probability weighted estimator (IPWE) \citep{robin2000marginal} or the augmented inverse probability weighted estimator (AIPWE) \citep{robins1994estimation}, where
\begin{align*}
\hat{V}^{\mathrm{IPWE}}_{\mu}(\hat{g})=&\frac{\sum_{i=1}^n\identity{\{A_i=\hat{g}(\Xbf_i)\}}Y_i/p(A_i|\Xbf_i)}
{\sum_{i=1}^n\identity{\{A_i=\hat{g}(\Xbf_i)\}}/p(A_i|\Xbf_i)},\\
\hat{V}^{\mathrm{AIPWE}}_{\mu}(\hat{g})=&\frac{1}{n}\sum_{i=1}^{n}\hat{\expect}(Y_i|\Xbf_i,\hat{g}(\Xbf_i))
+\frac{1}{n}\sum_{i=1}^{n}\frac{\identity{\{A_i=\hat{g}(\Xbf_i)\}}}{p(A_i|\Xbf_i)}\left[Y_i-\hat{\expect}(Y_i|\Xbf_i,A_i)\right],
\end{align*}
$\hat{\expect}(Y_i|\Xbf_i,A_i))=\varphi(\Xbf_i;\hat{\gammabf})+\left\{A_i-p(A_i|\Xbf_i)\right\}C(\Xbf_i;\hat{\betabf})$, and $p(A_i|\Xbf_i)\equiv0.5$. Both $\hat{V}^{\mathrm{IPWE}}_{\mu}(\hat{g})$ and $\hat{V}^{\mathrm{AIPWE}}_{\mu}(\hat{g})$ are consistent estimator of value $V_{\mu}(\hat{g})$, and their asymptotic covariance matrix can also be consistently estimated from the data \citep{zhang2012robust,mckeague2014estimation}. The estimates of $(V_{\mu}(\hat{g}))$ and their corresponding 95\% confidence interval of four methods based on both IPWE and AIPWE are given in Table~\ref{table:aids_value}.  

\begin{table}[ht]
	\setlength{\tabcolsep}{4pt}
	\caption{Result of estimated values and their corresponding 95\% confidence interval for four methods based on IPWE and AIPWE.SE stands for standard error. CI stands for 95\% confidence interval.}
	\label{table:aids_value}
	\begin{center}
		\begin{tabular}{llccc}
			\toprule
			Estimator & method & Value & SE & CI \\ 
			\midrule
			IPWE  & Least Square & 405.05 & 6.72 & (391.88, 418.22) \\ 
			& Pinball(0.5) & 406.77 & 6.71 & (393.63, 419.92) \\ 
			& Pinball(0.25) & 406.07 & 6.73 & (392.87, 419.26) \\ 
			& Huber & 407.03 & 6.71 & (393.87, 420.18) \\ 
			AIPWE & Least Square & 404.39 & 6.12 & (392.40, 416.38) \\ 
			& Pinball(0.5) & 405.93 & 6.13 & (393.92, 417.94) \\ 
			& Pinball(0.25) & 403.60 & 6.62 & (390.62, 416.58) \\ 
			& Huber & 406.00 & 6.15 & (393.95, 418.04) \\ 
			\bottomrule
		\end{tabular}
	\end{center}
\end{table}

From Table~\ref{table:aids_value}, robust regression with $\rho_{0.5}$ and Huber loss perform slightly better than lsA-learning, while robust regression with $\rho_{0.25}$ performs worse than lsA-learning when the values $(V_{\mu}(\hat{g}))$ is estimated based on AIPWE. We conduct KCI-test to check the conditional independence assumption $\epsilon\perp A|\Xbf$. For $\mathrm{RR}(\rho(0.5))$, $\mathrm{RR}(\rho(0.25))$ and RR(H), their p-values associated with KCI-test are 0.060, 0.002 and 0.083 respectively. The conditional independence assumption holds at the significance level of 0.05 for $\mathrm{RR}(\rho(0.5))$ and RR(H), so the estimated ITR can be thought to maximize $V_{\mu}(g)$. On the other hand, this assumption doesn't hold for $\mathrm{RR}(\rho(0.25))$, and its estimated ITR doesn't maximize $V_{\mu}(g)$, instead it approximately maximizes $V_{0.25-q}(g)$. This partly explains the relatively bad performance of RR($\rho_{0.25}$) in Table~\ref{table:aids_value}. Again, as $\mathrm{RR}(\rho(0.5))$ and RR(H) are more robust against heterogeneous, right skewed errors comparing with the least square method, they slightly outperform lsA-learning in term of $V_{\mu}(g)$.

\section{Discussion}

In this article, we propose a new general loss based robust regression framework for estimating the optimal individualized treatment rules. This new method has the desired property to be robust against skewed, heterogeneous, heavy-tailed errors and outliers. And similar as A-learning, it produces consistent estimates of the optimal ITR even when the baseline function is misspecified. However, the consistency of the proposed method does require the key conditional independence assumption $\epsilon\perp A|\Xbf$, which is somewhat stronger than the condition needed for the consistency of Q- and A-learning $(\expect(\epsilon|\Xbf,A)=0)$. So there are situations when the classical Q- and A-learning are more appropriate to apply. Furthermore, we also point out in the article that when pinball loss $\rho_{\tau}$ is chosen and the assumption $\epsilon\perp A|\Xbf$ doesn't hold, the estimated ITR approximately maximize the conditional $\tau$-th quantile and thus maximize $V_{\tau-q}(g)$. From a practice point of view, there are situations when maximizing $V_{\tau-q}(g)$ is a much more reasonable approach comparing with maximizing $V_{\mu}(g)$, especially when the conditional distribution of response $Y$ is highly skewed to one side.

In practice, there are cases when multiple treatment groups need to be compared simultaneously. For brevity, we have limited our discussion to two treatment groups. However, the proposed method
can be readily extended to multiple cases by just replacing equation \eqref{eq:A-general-loss} with the following more complex form,
\begin{equation*}
L_{3n}(\betabf, \gammabf)=\frac{1}{n}\sum_{i=1}^{n}M\left[Y_i-\varphi(\Xbf_i;\gammabf)-\sum_{k=1}^{K-1}(I(A_i=k)-\pi_k(\Xbf_i))C_k(\Xbf_i;\betabf_k)\right],
\label{eq:A-general-loss2}
\end{equation*}
where $\mathcal{A}=\{1,\ldots,K\}$, $K$-th treatment is the baseline treatment, $\pi_k(\Xbf_i)=\Pr(A_i=k|\Xbf_i)$ and $C_k(\Xbf_i;\betabf_k)$ denotes the contrast function comparing $k$-th treatment and the baseline treatment. All Theorems can be easily extended to this multiple treatments setting as well.

When the dimension of prognostic variables is high, regularized regression is needed in order to produce parsimonious yet interpretable individualized treatment rules. Essentially this is a variable selection problem in the context of M-estimator, which has been previously studied in \cite*{wu2009variable,li2011nonconcave}, etc. This is an interesting topic that needs further investigation. Another interesting direction is to extend the current method to the multi-stage setting, where sequential decisions are made along the time line.

\section*{Appendix A: Proof of Asymptotic Properties}
We consider the following additive model,
\begin{equation*}
Y_i=\varphi_0(\Xbf_i)+\{A_i-\pi(\Xbf_i)\}C(\Xbf_i;\beta_0)+\epsilon_i,\; i=1,\ldots,n,
\end{equation*}
where $\varphi_0(\Xbf)$ is the baseline function, $C(\Xbf;\beta_0)$ is the contrast function, $\pi(\Xbf)$ is the propensity score, and $\epsilon$ is the error term. We estimate $(\betabf, \gammabf)$ by minimizing
\begin{equation}
L_{3n}(\betabf, \gammabf)=\frac{1}{n}\sum_{i=1}^{n}M\left[Y_i-\varphi(\Xbf_i;\gammabf)-\{A_i-\pi(\Xbf_i)\}C(\Xbf_i;\betabf)\right],
\label{eq:A-general-loss}
\end{equation}
where $\gammabf\in\Gamma$, $\betabf\in\mathcal{B}$ and $M:\real \rightarrow [0,\infty)$ is a convex function with minimum achieved at 0. We consider the following three types of loss functions, i.e., the pinball loss
\begin{equation*}
M(x)=\rho_\tau(x)\triangleq
\begin{cases}
(\tau-1)x, &\text{if } x<0\\
\tau x, &\text{if } x\geq0
\end{cases}
\end{equation*}
where $0<\tau<1$, the Huber loss
\begin{equation*}
M(x)=H_\alpha(x)\triangleq
\begin{cases}
0.5x^2, &\text{if } |x|<\alpha\\
\alpha|x|-0.5\alpha^2, &\text{if } |x|\geq\alpha
\end{cases}
\end{equation*}
for some $\alpha>0$, and the $\epsilon$-insensitive loss
\begin{equation*}
M(x)=J_\epsilon(x)\triangleq\max(0, |x|-\epsilon)
\end{equation*}
for some $\epsilon>0$. Define $\Delta C(\xbf;\betabf)=C(\xbf;\betabf)-C(\xbf;\betabf_0)$. Assume $\gammabf\in\Gamma$, $\betabf\in\mathcal{B}$ and $\gammabf'$ is any arbitrary fix point in $\Gamma$.


\noindent \textbf{Regularity conditions A:}
\begin{itemize}
	\item[(A1)] $\{(Y_i,\Xbf_i,A_i,\epsilon_i),i=1,\ldots,n\}$ are i.i.d random variables.
	\item[(A2)] $\epsilon_i\perp A_i|\Xbf_i$ $\forall i=1,\ldots,n$.
	\item[(A3)] $\expect|\Delta C(\Xbf_i;\betabf)|<\infty$ $\forall\betabf\in\mathcal{B}$.
	\item[(A4)] $\Pr\{\xbf\in\mathcal{X}:\;\Delta C(\xbf;\betabf)\neq 0\}>0$ for all $\betabf\neq\betabf_0$.
	\item[(A5)] $\expect|\varphi(\Xbf_i;\gammabf)|<\infty$ $\forall\gammabf\in\Gamma$.
	\item[(A6)] $G_2(\gammabf)$ has unique minimizer $\gammabf^*$, where $G_2(\gammabf)$ is the pointwise limit of $L_{3n}(\betabf_0,\gammabf)-L_{3n}(\betabf_0,\gammabf')$ in probability.
	\item[(A7)] $L_{3n}(\betabf,\gammabf)$ is strictly convex with respect to $(\betabf,\gammabf)$.
	\item[(A8)] $\epsilon|\Xbf=\xbf$ has nonzero density on $\mathbb{R}$ for almost all $\xbf\in\mathcal{X}$.
\end{itemize}

\begin{lemma}
	$\left|\rho_\tau(x-y)-\rho_\tau(x)\right|\leq |y|$, for all $\tau\in(0,1)$.
\end{lemma}
\begin{proof}
	\begin{align*}
	\left|\rho_\tau(x-y)-\rho_\tau(x)\right| &= \left|\tau\left\{(x-y)_{+}-x_{+}\right\}+(1-\tau)\left\{(x-y)_{-}-x_{-}\right\}\right|\\
	&\leq|(x-y)_{+}-x_{+}|+|(x-y)_{-}-x_{-}|=|y|
	\end{align*}
\end{proof}

\begin{lemma}
	\begin{align*}
	\rho_\tau(x-y)-\rho_\tau(x)=&-\tau y\identity\{x\geq 0\}+(1-\tau)y\identity\{x< 0\}+(y-x)\identity\{x\geq 0\}\identity\{y>x\}\\
	&+(x-y)\identity\{x< 0\}\identity\{y< x\},
	\end{align*}
	for all $\tau\in(0,1)$.
\end{lemma}

\begin{proof}
	Denote $D=\rho_\tau(x-y)-\rho_\tau(x)$.
	\begin{enumerate}
		\item If $x\geq0$, $y\leq0$ $\Rightarrow$ $D=-\tau y$;
		\item If $x\geq0$, $y>0$, $|x|\geq|y|$ $\Rightarrow$ $D=-\tau y$;
		\item If $x\geq0$, $y>0$, $|x|<|y|$ $\Rightarrow$ $D=-\tau y+(y-x)$;
		\item If $x<0$, $y\geq0$ $\Rightarrow$ $D=(1-\tau)y$;
		\item If $x<0$, $y<0$, $|x|\geq|y|$ $\Rightarrow$ $D=(1-\tau)y$;
		\item If $x<0$, $y<0$, $|x|<|y|$ $\Rightarrow$ $D=(1-\tau)y+(x-y)$;
	\end{enumerate}
	Combining the above 6 cases, Lemma 2 is proved.
\end{proof}

\noindent \textbf{Proof of Theorem 1.}
\begin{proof} Recall that the loss function defined in \eqref{eq:A-general-loss} takes the form
	\begin{equation*}
	L_{3n}(\betabf, \gammabf)=\frac{1}{n}\sum_{i=1}^{n}\rho_\tau\left[\varphi_0(\Xbf_i)-\varphi(\Xbf_i;\gammabf)+\epsilon_i
	-(A_i-\pi(\Xbf_i))\Delta C(\Xbf_i;\betabf)\right].
	\end{equation*}
	By definition,
	\begin{align*}
	(\hat{\betabf}^{R}_{\rho(\tau)},\hat{\gammabf}^{R}_{\rho(\tau)}) =& \argmin_{(\betabf,\gammabf)}L_{3n}(\betabf,\gammabf)-L_{3n}(\betabf_0,\gammabf')\\
	=& \argmin_{(\betabf,\gammabf)}\left[L_{3n}(\betabf,\gammabf)-L_{3n}(\betabf_0,\gammabf)\right]+
	\left[L_{3n}(\betabf_0,\gammabf)-L_{3n}(\betabf_0,\gammabf')\right],
	\end{align*}
	Define
	\begin{align*}
	S_{1n}(\betabf, \gammabf) =& L_{3n}(\betabf, \gammabf)-L_{3n}(\betabf_0, \gammabf)=1/n\sum_{i=1}^n d_{1i};\\
	S_{2n}(\betabf, \gammabf) =& L_{3n}(\betabf_0, \gammabf)-L_{3n}(\betabf_0, \gammabf')=1/n\sum_{i=1}^n d_{2i}
	\end{align*}
	where
	\begin{align*}
	d_{1i} =& \rho_\tau\left[\varphi_0(\Xbf_i)-\varphi(\Xbf_i;\gammabf)+\epsilon_i
	-(A_i-\pi(\Xbf_i))\Delta C(\Xbf_i;\betabf)\right]-\rho_\tau\left[\varphi_0(\Xbf_i)-\varphi(\Xbf_i;\gammabf)+\epsilon_i\right],\\
	d_{2i} =& \rho_\tau\left[\varphi_0(\Xbf_i)-\varphi(\Xbf_i;\gammabf)+\epsilon_i\right]-\rho_\tau\left[\varphi_0(\Xbf_i)-\varphi(\Xbf_i;\gammabf')+\epsilon_i\right].
	\end{align*}
	By Lemma 1, A3 and A5, $\expect|d_{1i}|\leq\expect|(A_i-\pi(\Xbf_i))\Delta C(\Xbf_i;\betabf)|\leq\expect|\Delta C(\Xbf_i;\betabf)|<\infty$ and $\expect|d_{2i}|\leq\expect|\varphi(\Xbf_i;\gammabf)-\varphi(\Xbf_i;\gammabf')|\leq\expect|\varphi(\Xbf_i;\gammabf)|+\expect|\varphi(\Xbf_i;\gammabf')|
	<\infty$. Then, by Law of Large Number, $\forall\;\betabf\in\mathcal{B}$, $\gammabf\in\Gamma$, we have $S_{1n}(\betabf, \gammabf)\inprob G_1(\betabf,\gammabf)\triangleq\expect (D)$, and $S_{2n}(\betabf, \gammabf)\inprob G_2(\gammabf)$, where
	\begin{align*}
	D=&\rho_\tau\left[\varphi_0(\Xbf)-\varphi(\Xbf;\gammabf)+\epsilon
	-\{A-\pi(\Xbf)\}\Delta C(\Xbf;\betabf)\right]-\rho_\tau\left[\varphi_0(\Xbf)-\varphi(\Xbf;\gammabf)+\epsilon\right].
	\end{align*}
	Below we show that a) $(\betabf_0,\gammabf^*)$ is the minimizer of $G_1(\betabf,\gammabf)+G_2(\gammabf)$, b) $(\betabf_0,\gammabf^*)$ is the unique minimizer. The consistency then follows from the argmax continuous mapping theorem under Assumption (A7).
	
	Denote $K_1=\varphi_0(\Xbf)-\varphi(\Xbf;\gammabf)+\epsilon$, $K_2=\{A-\pi(\Xbf)\}\Delta C(\Xbf;\betabf)$. By Lemma 2,
	\begin{align*}
	D =& -\tau K_2\identity\{K_1\geq0\}+(1-\tau) K_2\identity\{K_1<0\}+(K_2-K_1)\identity\{K_1\geq0\}\identity\{K_2>K_1\}\\
	&+(K_1-K_2)\identity\{K_1<0\}\identity\{K_2<K_1\}.
	\end{align*}
	Since $\epsilon\perp A|\Xbf$ and $\Pr(A|\Xbf)=\pi(\Xbf)$, applying double expectation rule with $\Xbf$, we have $\expect[-\tau K_2\identity\{K_1\geq0\}]=\expect[(1-\tau) K_2\identity\{K_1<0\}]=0$.
	Thus,
	\begin{equation}
	G_1(\betabf,\gammabf)=\expect[(K_2-K_1)\identity\{K_1\geq0\}\identity\{K_2>K_1\}]+\expect[(K_1-K_2)\identity\{K_1<0\}\identity\{K_2<K_1\}].
	\label{eq:Gfunction}
	\end{equation}
	It is easy to check $G_1(\betabf,\gammabf)\geq0$ and achieves minimal value 0 at point $(\betabf_0,\gammabf)$ for all $\gammabf\in\Gamma$. In addition, by A6, we know $G_2(\gammabf)$ has unique minimizer $\gammabf^*$. Combining the above two facts, a) is proved.
	
	Combining A4, A8 and \eqref{eq:Gfunction}, we could prove $G_1(\betabf,\gammabf)>0$ for all $\betabf\neq\betabf_0$ and $\gammabf\in\Gamma$. So b) holds.
\end{proof}

\noindent \textbf{Proof of Theorem 3.}
\begin{proof} (a) When $M(x)=H_{\alpha}(x)$, the proof follows similar steps as Theorem 1. The only difference is that $G_1(\betabf,\gammabf)$ takes a different expression now and we need to redo the proof of 1) $G_1(\betabf,\gammabf)>0$ $\forall\betabf\neq\betabf_0$, $\gammabf\in\Gamma$, and 2) $G_1(\betabf_0,\gammabf)=0$ $\forall\gammabf\in\Gamma$. By definition, $G_1(\betabf,\gammabf)\triangleq\expect (D)$, where
	\begin{align*}
	D=H_\alpha\left[\varphi_0(\Xbf)-\varphi(\Xbf;\gammabf)+\epsilon
	-\{A-\pi(\Xbf)\}\Delta C(\Xbf;\betabf)\right]-H_\alpha\left[\varphi_0(\Xbf)-\varphi(\Xbf;\gammabf)+\epsilon\right].
	\end{align*}
	Then, 2) holds immediately. Denote $K_1=\varphi_0(\Xbf)-\varphi(\Xbf;\gammabf)+\epsilon$, $K_2=\{A-\pi(\Xbf)\}\Delta C(\Xbf;\betabf)$. We have the following four cases:
	\begin{enumerate}
		\item If $K_1>\alpha$ then $H_\alpha(K_1-K_2)\geq\alpha(K_1-K_2)-0.5\alpha^2$. Thus, $D\geq-\alpha K_2$;
		\item If $K_1<-\alpha$ then $H_\alpha(K_1-K_2)\geq\alpha(K_2-K_1)-0.5\alpha^2$. Thus, $D\geq\alpha K_2$;
		\item If $K_1\in[-\alpha,\alpha]$ and $K_1-K_2\in[-\alpha,\alpha]$ then $D=1/2(K_1-K_2)^2-1/2K_1^2=-K_1K_2+1/2K_2^2$;
		\item If $K_1\in[-\alpha,\alpha]$ and $K_1-K_2\not\in[-\alpha,\alpha]$ then $H_\alpha(K_1-K_2)\geq1/2(K_1-K_2)^2-\left\{1/2(\alpha+|K_2|)^2-\left[\alpha(\alpha+|K_2|)-1/2\alpha^2\right]\right\}=1/2(K_1-K_2)^2-1/2K_1^2$. Thus, $D\geq1/2(K_1-K_2)^2-1/2K_1^2-1/2K_2^2=-K_1K_2$.
	\end{enumerate}
	Combining the above four equalities and inequalities,
	\begin{align*}
	G_1(\betabf,\gammabf)\geq& \expect[-\alpha K_2\identity\{K_1>\alpha\}] + \expect[\alpha K_2\identity\{K_1<-\alpha\}] +
	\expect[-K_1K_2\identity\{K_1\in[-\alpha,\alpha]\}]\\
	&+ \expect\left[1/2K_2^2\identity\left(\{K_1\in[-\alpha,\alpha]\}\cup\{K_1-K_2\in[-\alpha,\alpha]\}\right)\right]
	\end{align*}
	Since $\epsilon\perp A|\Xbf$ and $\Pr(A|\Xbf)=\pi(\Xbf)$, applying double expectation rule with $\Xbf$, we have $\expect[-\alpha K_2\identity\{K_1>\alpha\}]=\expect[\alpha K_2\identity\{K_1<-\alpha\}]=\expect[-K_1K_2\identity\{K_1\in[-\alpha,\alpha]\}]=0$.
	Thus,
	\begin{equation}
	G_1(\betabf;\gammabf)\geq\expect\left[1/2K_2^2\identity\left(\{K_1\in[-\alpha,\alpha]\}\cup\{K_1-K_2\in[-\alpha,\alpha]\}\right)\right].
	\label{eq:Gfunction_thm2}
	\end{equation}
	Combining \eqref{eq:Gfunction_thm2}, A4 and A8, we can check that 1) holds. Thus, part (a) is proved.
	
	(b) When $M(x)=J_{\epsilon}(x)$, similarly $D=J_\epsilon\left(K_1-K_2\right)-J_\epsilon\left(K_1\right)$. Notice that we have the following three cases:
	\begin{enumerate}
		\item If $K_1>\epsilon$ then $D\geq -K_2$;
		\item If $K_1<-\epsilon$ then $D\geq K_2$;
		\item If $K_1\in[-\epsilon,\epsilon]$ then $D\geq 0$;
	\end{enumerate}
	The rest of the proof follows similar steps as part (a).
\end{proof}

\noindent \textbf{Proof of Theorem 5.}
\begin{proof} From Theorem 1, $\betabf_{\tau}=\betabf_0$. Plugging this into Theorem 4 and applying double expectation rules, we have
	\begin{equation*}
	J(\tau)=\expect\left[f_{\epsilon}\left(\tilde{\Xbf}\trans\gammabf(\tau)-\varphi_0(\Xbf)|\Xbf\right)
	\left(\begin{array}{cc}
	\pi(\Xbf)\{1-\pi(\Xbf)\}\tilde{\Xbf}\tilde{\Xbf}\trans & \zerobf\\
	\zerobf & \tilde{\Xbf}\tilde{\Xbf}\trans
	\end{array}
	\right)\right]
	\end{equation*}
	and
	\begin{equation*}
	\Sigma(\tau,\tau)=\expect\left\{\left[\tau-\identity\left\{\epsilon<\tilde{\Xbf}\trans\gammabf(\tau)-\varphi_0(\Xbf)\right\}\right]^2
	\left(\begin{array}{cc}
	\pi(\Xbf)\{1-\pi(\Xbf)\}\tilde{\Xbf}\tilde{\Xbf}\trans & \zerobf\\
	\zerobf & \tilde{\Xbf}\tilde{\Xbf}\trans
	\end{array}
	\right)\right\}.
	\end{equation*}
	Thus, $\sqrt{n}(\hat{\betabf}(\tau)-\betabf_0)\indist N(\zerobf, J_{11}^{-1}(\tau)\Sigma_{11}(\tau,\tau)J_{11}^{-1}(\tau))$, where $J_{11}^{-1}(\tau)$ and $\Sigma_{11}(\tau,\tau)$ are defined as in Theorem 5. Conditional on $\Xbf$, $\identity\left\{\epsilon<\tilde{\Xbf}\trans\gammabf(\tau)-\varphi_0(\Xbf)\right\}$ is a binomial random variable with $p=\Pr\left(\epsilon<\tilde{\Xbf}\trans\gammabf(\tau)-\varphi_0(\Xbf)\right)$. Then, $\expect\left\{\left[\tau-\identity\{\epsilon<\tilde{\Xbf}\trans\gammabf(\tau)-\varphi_0(\Xbf)\}\right]^2|\Xbf\right\}=(p-\tau)^2+p(1-p)\leq \tau^2+|1-2\tau|$. Thus, $\Sigma_{11}(\tau,\tau)\leq\left(\tau^2+|1-2\tau|\right)\expect\left[\pi(\Xbf)\{1-\pi(\Xbf)\}\tilde{\Xbf}\tilde{\Xbf}\trans\right]$.
\end{proof}

\section*{Appendix B: Additional Simulation Results}
We conducted additional simulations with non-constant propensity scores. Specifically, we considered the following examples. 

\subsection*{Examples with error terms independent with treatment}
We consider the following two models with p=3,
\begin{itemize}
	\item Model I:
	\begin{equation*}
	Y_i=1+(X_{i1}-X_{i2})(X_{i1}+X_{i3})+\{A_i-\pi(\Xbf_{i})\}\betabf_0\trans\tilde{\Xbf}_i+\sigma(\Xbf_{i})\epsilon_i,
	\end{equation*}
	where $\Xbf_{i}=(X_{i1},X_{i2},X_{i3})\trans$ are multivariate normal with mean 0, variance 1, and $\mathrm{Corr}(X_{ij},X_{ik})=0.5^{|j-k|}$, $\tilde{\Xbf}_i=(1,\Xbf_i\trans)\trans$ and $\betabf_0=(0,1,-1,1)\trans$.
	\item Model II:
	\begin{equation*}
	Y_i=\gammabf_0\trans\tilde{\Xbf}_i+\{A_i-\pi(\Xbf_{i})\}\betabf_0\trans\tilde{\Xbf}_i+\sigma(\Xbf_{i})\epsilon_i,
	\end{equation*}
	where $\gammabf_0\trans=(0.5,4,1,-3)$, and $\Xbf_{i}$, $\tilde{\Xbf}_{i}$ and $\betabf_0$ are the same as Model I.
\end{itemize}

We take linear forms for both the baseline and the contrast functions, where $\varphi(\Xbf;\gammabf)=\gammabf\trans\tilde{\Xbf}$ and $C(\Xbf;\betabf)=\betabf\trans\tilde{\Xbf}$. We assume the propensity scores $\pi(\cdot)$ are known, and we study the non-constant case $(\pi(\Xbf_i)=\mathrm{logit}(\Xbf_{i1}-\Xbf_{i2}))$ here. In addition, We consider two different $\sigma(\Xbf_{i})$ functions, i.e., the homogeneous case with $\sigma(\Xbf_{i})=1$, and the heterogenous case with $\sigma(\Xbf_{i})=0.5+(X_{i1}-X_{i2})^2$. The simulation results are given in Table~\ref{table:modelI_nonconstant} and Table~\ref{table:modelII_nonconstant}.

\begin{table}[ht]
	\setlength{\tabcolsep}{4pt}
	\caption{Summary result of Model I with non-constant propensity scores. LS stands for lsA-learning. P(0.5) stands for robust regression with pinball loss and parameter $\tau=0.5$. P(0.25) stands for robust regression with pinball loss and parameter $\tau=0.25$. Huber stands for robust regression with Huber loss, where parameter $\alpha$ is tuned automatically with R function rlm. Column $\delta_{0.5}$ is multiplied by 10.}
	\label{table:modelI_nonconstant}
	\begin{center}
		\begin{tabular}{@{}ll ccc ccc ccc}
			\toprule
			\multicolumn{11}{c}{{\textbf{Homogeneous Error}}}\\
			\hline
			& & \multicolumn{3}{c}{{{\bfseries Normal}}}
			& \multicolumn{3}{c}{{{\bfseries Log-Normal}}}
			& \multicolumn{3}{c}{{{\bfseries Cauchy}}}\\
			\cmidrule(r){3-5}
			\cmidrule(lr){6-8}
			\cmidrule(l){9-11}
			n & method & mse  & PCD & $\delta_{0.5}$ & mse  & PCD  & $\delta_{0.5}$ & mse  & PCD & $\delta_{0.5}$\\
			\cmidrule(r){3-5}
			\cmidrule(lr){6-8}
			\cmidrule(l){9-11}
			100 & LS & 1.70 (0.061) & 81.9 & 0.91 & 2.90 (0.114) & 77.6 & 1.34 &  & 59.3 & 3.61 \\
			& P(0.5) & 1.90 (0.069) & 80.1 & 1.09 & 2.13 (0.073) & 78.3 & 1.25 & 3.54 (0.128) & 75.7 & 1.57 \\
			& P(0.25) & 2.35 (0.080) & 78.2 & 1.33 & 1.95 (0.076) & 80.4 & 1.08 & 8.45 (0.431) & 69.8 & 2.28 \\
			& Huber & 1.51 (0.053) & 82.1 & 0.89 & 1.77 (0.065) & 80.6 & 1.02 & 3.67 (0.127) & 75.4 & 1.60 \\
			200 & LS & 0.77 (0.026) & 86.8 & 0.50 & 1.35 (0.045) & 82.2 & 0.91 &  & 59.2 & 3.63 \\
			& P(0.5) & 0.88 (0.028) & 85.5 & 0.60 & 1.00 (0.029) & 83.0 & 0.79 & 1.54 (0.050) & 81.1 & 1.00 \\
			& P(0.25) & 1.06 (0.035) & 84.5 & 0.68 & 0.83 (0.027) & 85.9 & 0.59 & 3.61 (0.143) & 74.7 & 1.70 \\
			& Huber & 0.68 (0.022) & 87.3 & 0.46 & 0.81 (0.025) & 85.2 & 0.62 & 1.58 (0.052) & 80.7 & 1.03 \\
			400 & LS & 0.39 (0.012) & 90.2 & 0.28 & 0.65 (0.020) & 86.9 & 0.48 &  & 58.0 & 3.79 \\
			& P(0.5) & 0.43 (0.013) & 89.3 & 0.32 & 0.47 (0.014) & 88.4 & 0.38 & 0.73 (0.022) & 86.5 & 0.51 \\
			& P(0.25) & 0.53 (0.016) & 88.5 & 0.38 & 0.41 (0.013) & 90.5 & 0.27 & 1.50 (0.049) & 81.7 & 0.96 \\
			& Huber & 0.34 (0.010) & 90.6 & 0.25 & 0.39 (0.012) & 89.6 & 0.30 & 0.72 (0.022) & 86.3 & 0.53 \\
			800 & LS & 0.18 (0.006) & 93.3 & 0.13 & 0.32 (0.010) & 90.2 & 0.27 &  & 58.3 & 3.75 \\
			& P(0.5) & 0.21 (0.007) & 92.7 & 0.15 & 0.24 (0.007) & 91.5 & 0.20 & 0.36 (0.011) & 90.3 & 0.27 \\
			& P(0.25) & 0.28 (0.009) & 92.4 & 0.17 & 0.21 (0.007) & 93.4 & 0.13 & 0.78 (0.026) & 86.9 & 0.50 \\
			& Huber & 0.16 (0.005) & 93.7 & 0.11 & 0.19 (0.006) & 92.6 & 0.15 & 0.37 (0.010) & 89.9 & 0.28 \\
			\hline
			\multicolumn{11}{c}{{\textbf{Heterogenous Error}}}\\
			\hline
			& & \multicolumn{3}{c}{{{\bfseries Normal}}}
			& \multicolumn{3}{c}{{{\bfseries Log-Normal}}}
			& \multicolumn{3}{c}{{{\bfseries Cauchy}}}\\
			\cmidrule(r){3-5}
			\cmidrule(lr){6-8}
			\cmidrule(l){9-11}
			n & method & mse  & PCD & $\delta_{0.5}$ & mse  & PCD  & $\delta_{0.5}$ & mse  & PCD & $\delta_{0.5}$\\
			\cmidrule(r){3-5}
			\cmidrule(lr){6-8}
			\cmidrule(l){9-11}
			100 & LS & 2.84 (0.111) & 78.2 & 1.33 & 9.96 (0.773) & 72.0 & 2.06 &  & 55.2 & 4.18 \\
			& P(0.5) & 2.01 (0.082) & 80.6 & 1.09 & 2.18 (0.080) & 79.2 & 1.21 & 4.18 (0.189) & 74.1 & 1.81 \\
			& P(0.25) & 2.91 (0.110) & 76.7 & 1.52 & 3.22 (0.105) & 74.2 & 1.76 & 10.62 (0.475) & 65.3 & 2.87 \\
			& Huber & 1.90 (0.074) & 80.9 & 1.06 & 2.38 (0.090) & 78.1 & 1.32 & 5.06 (0.230) & 71.9 & 2.04 \\
			200 & LS & 1.46 (0.053) & 83.1 & 0.83 & 4.47 (0.371) & 76.8 & 1.51 &  & 56.3 & 4.04 \\
			& P(0.5) & 0.92 (0.033) & 86.4 & 0.55 & 0.98 (0.035) & 85.3 & 0.64 & 1.69 (0.065) & 81.5 & 0.98 \\
			& P(0.25) & 1.35 (0.049) & 83.3 & 0.81 & 1.47 (0.049) & 81.6 & 0.97 & 4.73 (0.241) & 71.9 & 2.05 \\
			& Huber & 0.86 (0.030) & 86.6 & 0.52 & 1.02 (0.036) & 84.7 & 0.68 & 2.11 (0.079) & 79.3 & 1.18 \\
			400 & LS & 0.74 (0.029) & 87.4 & 0.47 & 2.65 (0.402) & 81.4 & 1.04 &  & 56.2 & 4.06 \\
			& P(0.5) & 0.45 (0.016) & 90.2 & 0.29 & 0.44 (0.017) & 89.5 & 0.34 & 0.79 (0.029) & 87.2 & 0.49 \\
			& P(0.25) & 0.66 (0.025) & 88.3 & 0.41 & 0.70 (0.023) & 86.9 & 0.50 & 2.12 (0.091) & 79.5 & 1.19 \\
			& Huber & 0.43 (0.016) & 90.2 & 0.28 & 0.48 (0.018) & 89.0 & 0.36 & 1.01 (0.036) & 85.0 & 0.65 \\
			800 & LS & 0.36 (0.013) & 90.8 & 0.25 & 1.09 (0.066) & 85.0 & 0.69 &  & 56.3 & 4.02 \\
			& P(0.5) & 0.21 (0.008) & 93.2 & 0.14 & 0.24 (0.009) & 92.3 & 0.19 & 0.39 (0.014) & 90.5 & 0.27 \\
			& P(0.25) & 0.33 (0.013) & 91.7 & 0.21 & 0.36 (0.012) & 90.8 & 0.25 & 1.01 (0.034) & 84.9 & 0.65 \\
			& Huber & 0.20 (0.008) & 93.2 & 0.14 & 0.25 (0.009) & 92.1 & 0.19 & 0.49 (0.016) & 89.1 & 0.34 \\
			\bottomrule
		\end{tabular}
	\end{center}
\end{table}

\begin{table}[ht]
	\setlength{\tabcolsep}{4pt}
	\caption{Summary result of Model II with non-constant propensity scores. LS stands for lsA-learning. P(0.5) stands for robust regression with pinball loss and parameter $\tau=0.5$. P(0.25) stands for robust regression with pinball loss and parameter $\tau=0.25$. Huber stands for robust regression with Huber loss, where parameter $\alpha$ is tuned automatically with R function rlm. Column $\delta_{0.5}$ is multiplied by 10.}
	\label{table:modelII_nonconstant}
	\begin{center}
		\begin{tabular}{@{}ll ccc ccc ccc}
			\toprule
			\multicolumn{11}{c}{{\textbf{Homogeneous Error}}}\\
			\hline
			& & \multicolumn{3}{c}{{{\bfseries Normal}}}
			& \multicolumn{3}{c}{{{\bfseries Log-Normal}}}
			& \multicolumn{3}{c}{{{\bfseries Cauchy}}}\\
			\cmidrule(r){3-5}
			\cmidrule(lr){6-8}
			\cmidrule(l){9-11}
			n & method & mse  & PCD & $\delta_{0.5}$ & mse  & PCD  & $\delta_{0.5}$ & mse  & PCD & $\delta_{0.5}$\\
			\cmidrule(r){3-5}
			\cmidrule(lr){6-8}
			\cmidrule(l){9-11}
			100 & LS & 0.36 (0.011) & 89.8 & 0.29 & 1.65 (0.085) & 80.8 & 1.06 &  & 58.7 & 3.69 \\
			& P(0.5) & 0.57 (0.017) & 86.9 & 0.46 & 0.61 (0.026) & 86.4 & 0.55 & 1.31 (0.045) & 81.7 & 0.93 \\
			& P(0.25) & 0.65 (0.020) & 86.2 & 0.52 & 0.22 (0.008) & 91.7 & 0.20 & 4.67 (0.312) & 74.7 & 1.64 \\
			& Huber & 0.38 (0.012) & 89.5 & 0.30 & 0.45 (0.018) & 88.3 & 0.40 & 1.70 (0.060) & 79.5 & 1.14 \\
			200 & LS & 0.16 (0.004) & 92.9 & 0.14 & 0.74 (0.030) & 85.6 & 0.61 &  & 59.1 & 3.64 \\
			& P(0.5) & 0.25 (0.007) & 91.2 & 0.21 & 0.26 (0.008) & 90.7 & 0.24 & 0.52 (0.017) & 87.8 & 0.41 \\
			& P(0.25) & 0.30 (0.008) & 90.3 & 0.26 & 0.09 (0.003) & 94.8 & 0.08 & 1.69 (0.074) & 81.3 & 0.92 \\
			& Huber & 0.17 (0.005) & 92.8 & 0.14 & 0.19 (0.006) & 92.2 & 0.17 & 0.70 (0.022) & 86.2 & 0.53 \\
			400 & LS & 0.08 (0.002) & 95.1 & 0.06 & 0.36 (0.013) & 89.7 & 0.30 &  & 58.0 & 3.79 \\
			& P(0.5) & 0.12 (0.003) & 93.8 & 0.10 & 0.12 (0.003) & 93.8 & 0.10 & 0.22 (0.006) & 91.6 & 0.19 \\
			& P(0.25) & 0.14 (0.004) & 93.3 & 0.12 & 0.04 (0.001) & 96.5 & 0.03 & 0.63 (0.021) & 86.5 & 0.49 \\
			& Huber & 0.08 (0.002) & 95.0 & 0.07 & 0.09 (0.002) & 94.8 & 0.07 & 0.30 (0.009) & 90.3 & 0.26 \\
			800 & LS & 0.04 (0.001) & 96.5 & 0.03 & 0.18 (0.006) & 92.3 & 0.16 &  & 58.2 & 3.76 \\
			& P(0.5) & 0.06 (0.002) & 95.6 & 0.05 & 0.06 (0.002) & 95.6 & 0.05 & 0.10 (0.003) & 94.4 & 0.09 \\
			& P(0.25) & 0.07 (0.002) & 95.3 & 0.06 & 0.02 (0.001) & 97.5 & 0.02 & 0.29 (0.009) & 90.6 & 0.23 \\
			& Huber & 0.04 (0.001) & 96.4 & 0.03 & 0.04 (0.001) & 96.3 & 0.04 & 0.14 (0.004) & 93.2 & 0.12 \\
			\hline
			\multicolumn{11}{c}{{\textbf{Heterogenous Error}}}\\
			\hline
			& & \multicolumn{3}{c}{{{\bfseries Normal}}}
			& \multicolumn{3}{c}{{{\bfseries Log-Normal}}}
			& \multicolumn{3}{c}{{{\bfseries Cauchy}}}\\
			\cmidrule(r){3-5}
			\cmidrule(lr){6-8}
			\cmidrule(l){9-11}
			n & method & mse  & PCD & $\delta_{0.5}$ & mse  & PCD  & $\delta_{0.5}$ & mse  & PCD & $\delta_{0.5}$\\
			\cmidrule(r){3-5}
			\cmidrule(lr){6-8}
			\cmidrule(l){9-11}
			100 & LS & 1.45 (0.059) & 82.9 & 0.85 & 8.53 (0.784) & 72.4 & 2.01 &  & 54.9 & 4.22 \\
			& P(0.5) & 0.94 (0.034) & 85.6 & 0.61 & 1.29 (0.058) & 83.3 & 0.86 & 2.27 (0.132) & 78.9 & 1.24 \\
			& P(0.25) & 1.46 (0.051) & 81.5 & 0.96 & 1.78 (0.071) & 78.2 & 1.30 & 7.88 (0.422) & 68.1 & 2.46 \\
			& Huber & 0.89 (0.034) & 86.1 & 0.57 & 1.46 (0.067) & 81.7 & 0.99 & 3.28 (0.157) & 75.1 & 1.65 \\
			200 & LS & 0.84 (0.035) & 86.6 & 0.53 & 3.85 (0.358) & 77.6 & 1.43 &  & 55.9 & 4.09 \\
			& P(0.5) & 0.44 (0.016) & 90.0 & 0.29 & 0.60 (0.024) & 89.0 & 0.39 & 0.87 (0.034) & 86.3 & 0.56 \\
			& P(0.25) & 0.69 (0.025) & 87.0 & 0.49 & 0.75 (0.024) & 85.5 & 0.59 & 3.08 (0.179) & 75.3 & 1.58 \\
			& Huber & 0.43 (0.016) & 90.3 & 0.28 & 0.66 (0.025) & 87.7 & 0.47 & 1.32 (0.050) & 82.4 & 0.87 \\
			400 & LS & 0.44 (0.020) & 90.3 & 0.28 & 2.34 (0.393) & 82.4 & 0.95 &  & 55.9 & 4.09 \\
			& P(0.5) & 0.23 (0.009) & 92.9 & 0.16 & 0.28 (0.011) & 92.5 & 0.19 & 0.39 (0.015) & 90.8 & 0.26 \\
			& P(0.25) & 0.33 (0.011) & 91.0 & 0.23 & 0.36 (0.012) & 90.1 & 0.27 & 1.25 (0.048) & 82.8 & 0.82 \\
			& Huber & 0.22 (0.008) & 93.1 & 0.15 & 0.31 (0.012) & 91.7 & 0.21 & 0.60 (0.022) & 88.0 & 0.43 \\
			800 & LS & 0.23 (0.009) & 93.0 & 0.15 & 0.90 (0.057) & 86.2 & 0.60 &  & 56.3 & 4.03 \\
			& P(0.5) & 0.11 (0.004) & 95.0 & 0.07 & 0.14 (0.005) & 94.8 & 0.09 & 0.18 (0.006) & 93.6 & 0.12 \\
			& P(0.25) & 0.17 (0.006) & 93.7 & 0.12 & 0.18 (0.006) & 93.0 & 0.14 & 0.59 (0.017) & 87.3 & 0.44 \\
			& Huber & 0.10 (0.004) & 95.1 & 0.07 & 0.15 (0.006) & 94.2 & 0.11 & 0.29 (0.010) & 91.4 & 0.21 \\
			\bottomrule
		\end{tabular}
	\end{center}
\end{table}

We firstly notice that lsA-learning works much worse under the heterogeneous errors, while all other methods are generally less affected by the heterogeneity of the errors. When the baseline function is misspecified as in Model I, under the homogeneous normal errors, RR(H) works slightly better than lsA-learning, while $\mathrm{RR}(\rho_{0.25})$ works the worst. The difference in general is small. For the homogeneous log-normal errors, again RR(H) works the best, while $\mathrm{RR}(\rho_{0.5})$ and $\mathrm{RR}(\rho_{0.25})$ work slightly worse. Here lsA-learning has the worst performance. Under the homogeneous Cauchy errors, the lsA-learning is no longer consistent and work the worst. Both $\mathrm{RR}(\rho_{0.5})$ and RR(H) have good performance under the homogeneous Cauchy errors. When baseline function is correctly specified as in Model II, under homogeneous normal errors, lsA-learning performs the best. However, in this case RR(H) also has a very close performance. Under homogeneous log-normal errors, $\mathrm{RR}(\rho_{0.25})$ work the best and lsA-learning work the worst. Under homogeneous Cauchy errors, $\mathrm{RR}(\rho_{0.5})$ has the best performance and RR(H) has a close performance. lsA-learning is again not consistent. 

\subsection*{Examples with error terms interacted with treatment}
We consider the following model with p=2,
\begin{equation*}
Y_i=1 + 0.5\sin[\pi(X_{i1}-X_{i2})]+
0.25(1+X_{i1}+2X_{i2})^2+(A_i-\pi(\Xbf_{i}))\thetabf_0\trans\tilde{\Xbf}_i+\sigma(\Xbf_{i},A_i)\epsilon_i,
\end{equation*}
where $\Xbf_{i}=(X_{i1},X_{i2})\trans$, $\tilde{\Xbf}_i=(1,\Xbf_i\trans)\trans$, $\sigma(\Xbf_{i},A_i)=1+A_i d_0 X_{i1}^2$, $\thetabf_0\trans=(0.5,2,-1)$ and $X_{ik}$ are i.i.d. Uniform[-1,1]. We take linear forms for both the baseline and the contrast functions, where $\varphi(\Xbf;\gammabf)=\gammabf\trans\tilde{\Xbf}$, $C(\Xbf;\betabf)=\betabf\trans \Wbf$ and $\Wbf=(\tilde{\Xbf},X_{1}^2,X_{2}^2,X_{1}X_{2})$. $d_0=5$, 10 or 15. The error terms $\epsilon_i$ follows i.i.d. N(0,1) or Gamma(1,1)-1 distribution. The propensity scores $\pi(\cdot)$ are known, and we consider the non-constant case ($\pi(\Xbf_i)=\mathrm{logit}(\Xbf_{i1}-\Xbf_{i2})$) here. The simulation results are given in Table~\ref{table:interacted_nonconstant_ps}.

\begin{table}[ht]
	\setlength{\tabcolsep}{2.5pt}
	\caption{Summary results with non-constant propensity scores when errors interacted with treatment. Least square stands for lsA-learning. Pinball(0.5) stands for robust regression with pinball loss and parameter $\tau=0.5$. Pinball(0.25) stands for robust regression with pinball loss and parameter $\tau=0.25$.  Huber stands for robust regression with Huber loss, where parameter $\alpha$ is tuned automatically with R function rlm.}
	\label{table:interacted_nonconstant_ps}
	\begin{center}
		\begin{tabular}{@{}lll @{\hskip 0.5cm} rrr @{\hskip 0.5cm} rrr @{\hskip 0.5cm} rrr  @{\hskip 0.5cm} rrr}
			\toprule
			& & & \multicolumn{3}{c}{{{\bfseries Least Square}}} 
			& \multicolumn{3}{c}{{{\bfseries Pinball(0.5)}}} 
			& \multicolumn{3}{c}{{{\bfseries Pinball(0.25)}}} 
			& \multicolumn{3}{c}{{{\bfseries Huber}}}\\
			\cmidrule(r){4-6}
			\cmidrule(lr){7-9}
			\cmidrule(lr){10-12}
			\cmidrule(lr){13-15}
			Error & $d_0$ & n & $ \delta_{\mu}$ & $ \delta_{0.5}$ & $ \delta_{0.25}$ & $ \delta_{\mu}$ & $ \delta_{0.5}$ & $ \delta_{0.25}$ & $ \delta_{\mu}$ & $ \delta_{0.5}$ & $ \delta_{0.25}$ 
			& $ \delta_{\mu}$ & $ \delta_{0.5}$ & $ \delta_{0.25}$ \\
			\midrule
			Normal & 5 & 100 & 0.19 & 0.19 & 0.36 & 0.20 & 0.20 & 0.34 & 0.30 & 0.30 & 0.23 & 0.17 & 0.17 & 0.33 \\ 
			&  & 200 & 0.11 & 0.11 & 0.28 & 0.13 & 0.13 & 0.25 & 0.21 & 0.21 & 0.12 & 0.11 & 0.11 & 0.23 \\ 
			&  & 400 & 0.06 & 0.06 & 0.21 & 0.08 & 0.08 & 0.17 & 0.17 & 0.17 & 0.06 & 0.06 & 0.06 & 0.16 \\ 
			&  & 800 & 0.03 & 0.03 & 0.16 & 0.06 & 0.06 & 0.11 & 0.15 & 0.15 & 0.04 & 0.04 & 0.04 & 0.11 \\ 
			& 10 & 100 & 0.29 & 0.29 & 0.93 & 0.24 & 0.24 & 0.88 & 0.44 & 0.44 & 0.50 & 0.24 & 0.24 & 0.88 \\ 
			&  & 200 & 0.21 & 0.21 & 0.92 & 0.18 & 0.18 & 0.84 & 0.37 & 0.37 & 0.34 & 0.17 & 0.17 & 0.83 \\ 
			&  & 400 & 0.13 & 0.13 & 0.87 & 0.14 & 0.14 & 0.75 & 0.32 & 0.32 & 0.25 & 0.12 & 0.12 & 0.75 \\ 
			&  & 800 & 0.08 & 0.08 & 0.80 & 0.11 & 0.11 & 0.64 & 0.28 & 0.28 & 0.21 & 0.08 & 0.08 & 0.64 \\ 
			& 15 & 100 & 0.35 & 0.35 & 1.58 & 0.27 & 0.27 & 1.51 & 0.53 & 0.53 & 0.72 & 0.26 & 0.26 & 1.51 \\ 
			&  & 200 & 0.29 & 0.29 & 1.56 & 0.21 & 0.21 & 1.47 & 0.50 & 0.50 & 0.54 & 0.20 & 0.20 & 1.47 \\ 
			&  & 400 & 0.21 & 0.21 & 1.58 & 0.17 & 0.17 & 1.37 & 0.48 & 0.48 & 0.39 & 0.15 & 0.15 & 1.38 \\ 
			&  & 800 & 0.14 & 0.14 & 1.52 & 0.14 & 0.14 & 1.26 & 0.45 & 0.45 & 0.31 & 0.12 & 0.12 & 1.27 \\ 
			Gamma & 5 & 100 & 0.18 & 0.21 & 0.34 & 0.20 & 0.17 & 0.24 & 0.28 & 0.18 & 0.14 & 0.18 & 0.15 & 0.21 \\ 
			&  & 200 & 0.10 & 0.14 & 0.29 & 0.13 & 0.10 & 0.15 & 0.21 & 0.11 & 0.07 & 0.11 & 0.07 & 0.13 \\ 
			&  & 400 & 0.06 & 0.09 & 0.23 & 0.10 & 0.05 & 0.10 & 0.18 & 0.07 & 0.04 & 0.07 & 0.03 & 0.08 \\ 
			&  & 800 & 0.03 & 0.06 & 0.19 & 0.08 & 0.03 & 0.06 & 0.16 & 0.06 & 0.03 & 0.06 & 0.02 & 0.07 \\ 
			& 10 & 100 & 0.27 & 0.34 & 0.90 & 0.28 & 0.25 & 0.67 & 0.46 & 0.21 & 0.33 & 0.28 & 0.22 & 0.62 \\ 
			&  & 200 & 0.20 & 0.32 & 0.94 & 0.21 & 0.16 & 0.57 & 0.43 & 0.14 & 0.24 & 0.21 & 0.13 & 0.49 \\ 
			&  & 400 & 0.13 & 0.27 & 0.92 & 0.16 & 0.09 & 0.46 & 0.38 & 0.10 & 0.18 & 0.15 & 0.06 & 0.39 \\ 
			&  & 800 & 0.08 & 0.21 & 0.85 & 0.13 & 0.05 & 0.40 & 0.35 & 0.07 & 0.16 & 0.13 & 0.03 & 0.35 \\ 
			& 15 & 100 & 0.34 & 0.55 & 1.49 & 0.33 & 0.37 & 1.09 & 0.59 & 0.25 & 0.46 & 0.33 & 0.33 & 0.99 \\ 
			&  & 200 & 0.27 & 0.54 & 1.57 & 0.26 & 0.29 & 1.00 & 0.60 & 0.19 & 0.31 & 0.27 & 0.23 & 0.85 \\ 
			&  & 400 & 0.19 & 0.50 & 1.56 & 0.20 & 0.21 & 0.88 & 0.61 & 0.15 & 0.21 & 0.22 & 0.14 & 0.70 \\ 
			&  & 800 & 0.12 & 0.47 & 1.58 & 0.17 & 0.14 & 0.76 & 0.62 & 0.15 & 0.18 & 0.19 & 0.09 & 0.63 \\ 
			\bottomrule
		\end{tabular}
	\end{center}
\end{table}

Based on Theorem 6 of the main paper, $\delta_{\mu}$ column for the lsA-learning method in Table~\ref{table:interacted_nonconstant_ps} converges to 0 as sample size increases. Under Normal error terms, we have $\delta_{0.5}=\delta_{\mu}$. Thus, the $\delta_{0.5}$ column for the lsA-learning method under Normal error also converges to 0. All other columns in Table \ref{table:interacted_nonconstant_ps} converge to a positive constant instead of 0 as sample size goes to infinity. $\mathrm{RR}(H)$ and $\mathrm{RR}(\rho_{0.5})$ perform similarly in Table \ref{table:interacted_nonconstant_ps}. We also find even though lsA-learning outperform all other methods in $\delta_{\mu}$ when sample size is large. It may be worse than $\mathrm{RR}(\rho_{0.5})$ and $\mathrm{RR}(H)$ when sample size is small due to the fact that lsA-learning is inefficient under the heteroscedastic or skewed errors. Last, we find that lsA-learning, $\mathrm{RR}(\rho_{0.5})$  and $\mathrm{RR}(\rho_{0.25})$  perform best at the columns $\delta_{\mu}$, $\delta_{0.5}$ and $\delta_{0.25}$ accordingly. The reason is given in the Remark under Theorem 2 of the main paper.
\bibliography{OTR}
\bibliographystyle{apalike}

\end{document}